\documentclass[12pt]{article}
\usepackage{amssymb,amsmath}
\usepackage{color}
\numberwithin{equation}{section}
\newtheorem{theorem}{Theorem}[section]     
\newtheorem{definition}[theorem]{Definition}

\newtheorem{proposition}[theorem]{Proposition}
\newtheorem{lemma}[theorem]{Lemma}
\newtheorem{de}[theorem]{Definition}

\newtheorem{corollary}[theorem]{Corollary}

\newtheorem{remark}[theorem]{Remark}

\def\d{\partial}

\def\p{\partial}
\def\n{\noindent}
\def\f{\frac}

\def\proof{\noindent\hspace{2em}{\itshape Proof: }}
\def\QEDclosed{\mbox{\rule[0pt]{1.3ex}{1.3ex}}} 

\def\QED{\QEDclosed} 
\def\endproof{\hspace*{\fill}~\QED\par\endtrivlist\unskip}
\newcommand{\eqa}{\begin{eqnarray}}
\newcommand{\eeqa}{\end{eqnarray}}
\newcommand{\beq}{\begin{equation}}
\newcommand{\eeq}{\end{equation}}
\newcommand{\nn}{\nonumber}

\newcommand{\pal}{\partial}

\begin{document}

\title{Poisson bracket on $1$-forms and evolutionary partial differential
 equations}
\author{Alessandro Arsie* and Paolo Lorenzoni**\\
\\
{\small *Department of Mathematics and Statistics}\\
{\small University of Toledo,}
{\small 2801 W. Bancroft St., 43606 Toledo, OH, USA}\\
{\small **Dipartimento di Matematica e Applicazioni}\\
{\small Universit\`a di Milano-Bicocca,}
{\small Via Roberto Cozzi 53, I-20125 Milano, Italy}\\
{\small *alessandro.arsie@utoledo.edu,  **paolo.lorenzoni@unimib.it}}

\date{}

\maketitle

\begin{abstract}
We introduce a bracket on $1$-forms defined on ${\cal J}^{\infty}(S^1, \mathbb{R}^n)$, the infinite jet extension of the space of loops and prove that it satisfies the standard properties of a Poisson bracket. Using this bracket, we show that certain hierarchies appearing in the framework of $F$-manifolds with compatible flat connection $(M, \nabla, \circ)$ are Hamiltonian in a generalized sense. Moreover, we show that if a metric $g$ compatible with $\nabla$ is also invariant with respect to $\circ$, then this generalized Hamiltonian set-up reduces to the standard one. 
\end{abstract}

\section{Introduction}

In the study of integrable evolutionary PDEs, many geometric structures have been introduced, starting from the celebrated Dubrovin-Novikov framework relating the 
existence of a local Poisson tensor of hydrodynamic type to the presence of a flat non-degenerate metric $g$.
More precisely, in the Dubrovin-Novikov set-up given 
 two local functionals $F[u]=\int_{S^1}f(u,u_x,\dots)\,dx$ and $G[u]=\int_{S^1}g(u,u_x,\dots)\,dx$,
 their Poisson bracket is defined as
\beq\label{LPHT}
\{F,G\}=\int_{S^1}\f{\delta F}{\delta u^i}
\left(g^{ij}\d_x-g^{il}\Gamma_{lk}^ju^k_x\right)
\f{\delta G}{\delta u^j}\,dx
\eeq
where $\Gamma_{lk}^j$ are the Christoffel symbols
 of the Levi-Civita connection associated to the metric $g$.
\newline
\newline
\newline
Many important examples of evolutionary  PDEs are Hamiltonian with respect to a local Poisson bracket of hydrodynamic type.
 This means that they can be written as
\beq\label{Hform}
u^i_t=\left(g^{ij}\d_x-g^{il}\Gamma_{lk}^ju^k_x\right)
\f{\delta H}{\delta u^j},\qquad i=1,\dots,n
\eeq
for a suitable choice of the Hamiltonian $H[u]=\int_{S^1}h(u,u_x,\dots)\,dx$ (in the case of quasilinear equations
 the density $h$ depends only on $u$). In equation \eqref{Hform}, $\f{\delta H}{\delta u^j}$ denotes variational derivative of the functional $H$, see formula \eqref{var.der}. In this lucky case there is no need to extend the bracket to $1$-forms.
 
 However, the equations admitting this Hamiltonian formulation are  exceptions, even in the integrable case.
 For this reason, immediately after the seminal paper of Dudrovin and Novikov, many authors devoted lots of work
 to extend their formalism. These efforts led to the definition of non local Poisson bracket of hydrodynamic type \cite{fm90,fe91}.
 Besides the local part \eqref{LPHT} they contain an additional term of the form
\beq\label{LPHTcor}  
\sum_{\beta}\epsilon_{\beta}\left(w_{\beta}\right)^i_ku^k_x
\left(\frac{d}{dx}\right)^{\!-1}\!\!\!\left(w_{\beta}\right)^j_hu^h_x
\eeq
which is due to the non-vanishing of the curvature of $g$:
\beq\label{qexp}
R^{ij}_{kh}=\sum_\alpha\varepsilon_{\alpha} \left\{\left(w_\alpha\right)^i_k\left(w_\alpha\right)^j_h
-\left(w_\alpha\right)^j_k\left(w_\alpha\right)^i_h\right\}.
\eeq
The affinors $w^i_j$ appearing in the formula are related to symmetries and can be interpreted as Weingarten operators.
In practice, the presence of such term might have some drawbacks due to the ambiguity of the action of the
 operator $\d_x^{-1}$ on the differential of local functionals and to the difficulties
 in finding the quadratic expansion  \eqref{qexp}.
\newline
\newline
Recently it was observed that a special class of non-local Poisson structures of hydrodynamic type related to local Poisson
 structures by a reciprocal transformation can be interpreted as local Jacobi structure, see \cite{LZ}. 
 
 This discovery enabled the authors of \cite{LZ}
 to avoid the disadvantages created by the presence of the non-local term and it allowed them to compute the Lichnerowicz-Jacobi cohomology groups extending
 the results obtained by Getzler \cite{G} in the case of local Poisson structures.	

The results we present here go in the same direction but, instead of dropping the requirement of locality we drop the requirement
 of exactness of the 1-form $\omega$ defining the equation (or the hierarchy in the integrable case):
\beq\label{hameqne}
u^i_t=P^{ij}\omega_j=\left(g^{ij}\d_x-g^{il}\Gamma_{lk}^ju^k_x\right)
\omega_j,\qquad i=1,\dots,n.
\eeq

The motivation for this work comes from the study of integrable hierarchies related to $F$-manifold with compatible flat and bi-flat structure
 \cite{LPR,LP,AL1,AL2}. Both these classes of manifolds are equipped with a flat torsionless connection $\nabla$ that selects a class of local Poisson tensors $P$ of hydrodynamic type.

%
We show that the equations of these hierarchies can always be put in the form \eqref{hameqne}
 where the Poisson tensor $P$ belongs to the class defined by the connection $\nabla$, \footnote{By this we mean that $P$ is a Poisson structure of Dubrovin-Novikov type associated to a metric $g$ compatible with $\nabla$.} and $\omega$ is a suitable $1$-form. In the case of Frobenius manifolds, it turns out that the $1$-form $\omega$ is indeed exact, namely it can be written as $\omega=\delta H$, where the functional $H$ is interpreted as the Hamiltonian of the PDEs \eqref{hameqne}.  
In general however, these hierarchies do not admit any usual local Hamiltonian structure, namely the $1$-forms $\omega$ are not exact. 
 

The aim of this work is indeed to explore these issues and it is twofold. 

First of all it is natural to ask what is the relevance of these local Poisson tensors for the corresponding integrable hierarchies, when it is well-known that for most of them there is no corresponding Hamiltonian functional, so they can not be written in standard Hamiltonian form. We provide an answer introducing a Poisson bracket on $1$-forms, that are not necessarily closed (the case in which a Hamiltonian functional does exist correspond to {\em exact} $1$-forms). 

To do so we build on the work of \cite{GD, MM}. 
Indeed, let us recall that on any Poisson manifold $(M,P)$ there is 
 a $\mathbb{R}$-bilinear, skew symmetric operation 
 $\{\cdot,\cdot\}:\Lambda_1 M\times\Lambda_1 M\to
\Lambda_1 M$, where $\Lambda_1M$ denotes the vector space of $1$-forms on $M$,  extending the usual Poisson bracket
 between smooth functions. It is defined by
\begin{equation}\label{PB1intro}
\{\alpha,\beta\}:={\rm Lie}_{P\beta}\alpha-{\rm Lie}_{P\alpha}\beta+d\langle \beta,P\alpha\rangle.
\end{equation}
and satisfies
\begin{eqnarray*}
&&\{df,dg\}=d\{f,g\}\\
&&\{\alpha,f\beta\}=f\{\alpha,\beta\}+[(P\alpha)(f)]\beta.
\end{eqnarray*}
Here we consider an infinite-dimensional analogue of such a bracket. Our Poisson manifold will be the space
 $\mathcal{L}(M)=\{S^1\to\mathbb{R}^n\}$
 of $C^{\infty}$ maps from the circle to $\mathbb{R}^n$ endowed
 with a local Poisson bivector $P$ of hydrodynamic type.
 
 The Poisson bracket we are going to introduce on $1$-forms has the following expression. Let $\alpha=\int dx \wedge \alpha_i \delta u^i$ and $\beta=\int dx \wedge \beta_j \delta u^j$ be two $1$-forms. 
Then  $\{\alpha, \beta\}$ is a $1$-form $\{\alpha, \beta\}=\int dx\wedge \{\alpha, \beta\}_i\delta u^i$, 
where 
\beq\begin{split}\label{}
\{\alpha, \beta\}_i=\p_x^s\left(g^{kl}\p_x\beta_l+\Gamma^{kl}_mu^m_x\beta_l\right)\f{\p\alpha_i}{\p u^k_{(s)}}-\p_x^s\left(g^{kl}\p_x\alpha_l+\Gamma^{kl}_mu^m_x\alpha_l\right)\f{\p\beta_i}{\p u^k_{(s)}}\\
+\left(\alpha_k\p_x\beta_l-\beta_k\p_x\alpha_l\right)\Gamma^{lk}_i-\alpha_k\beta_l\left[\Gamma^{k}_{is}\Gamma^{sl}_m-\Gamma^l_{is}\Gamma^{sk}_m\right]u^m_x,
\end{split}\eeq
where $g^{kl}$ is the flat (contravariant)-metric in the chosen coordinates and $\Gamma^{lk}_m$ are the corresponding Christoffel symbols. 

Although it is not unexpected, one result about this bracket is that it satisfies Jacobi identity. Moreover, this bracket equips the vector space of $1$-forms with a Lie algebra structure  and the Poisson tensor $P$ gives rise to an anti-homomorphism of Lie algebras from the Lie algebra of $1$-forms equipped with this bracket and the Lie algebra of evolutionary vector fields equipped with the Lie bracket. 

Given an $F$-manifold with compatible flat connection $(M, \nabla, \circ)$, the second issue we are addressing is the role played by the existence of a metric $g$ invariant with respect to the product $\circ$ (namely $g(X\circ Y, Z)=g(X, Y\circ Z)$) and compatible with $\nabla$ ($\nabla g=0$) with respect to the nature of the $1$-forms we are considering. 

We prove that the invariance of the metric with respect to $\circ$ is responsible for the exactness of the $1$-forms, namely for the existence of a true Hamiltonian functional. 

However, in the general case of an $F$-manifold with  compatible flat connection $(M, \nabla, \circ)$ where such a metric does not exist, (see for instance \cite{AL2} where we provided plenty of examples in which any metric compatible with the connection $\nabla$ is {\em not invariant} with respect to the product $\circ$), it turns out that the $1$-forms on which the Poisson tensor acts are {\em not exact}.

In this context we provide also an alternative proof of the commutativity of the flows in the principal hierarchy. Indeed, in this set-up the commutativity of the flows is equivalent to the Poisson-involutivity of the corresponding $1$-forms with respect to the Poisson bracket introduced above. 
\newline
\newline
%
The paper is organized as follows. In Section 2 we introduce $1$-forms, evolutionary vector fields and operations on them, essentially following \cite{DZ}, to fix notation and for the sake of being self-contained. In Section 3 we introduce the Poisson bracket on $1$-forms and we derive its expression in flat and general coordinates. In Section 4 we prove the main properties of this bracket. 

In Section 5 we recall the concept of $F$-manifold with compatible flat connection and we show how the Poisson bracket just introduced fits in the description of the Hamiltonian structure of the corresponding principal hierarchy. In Section 6 we continue the exploration of the Hamiltonian structure of the principal hierarchy, in particular focusing on the role played by the invariance of the metric $g$ with respect to product $\circ$. 

In Section 7 we work out an example in dimension $3$.

\section{Differential forms on the formal loop space}

In this section, we are going to recall the construction of differential forms on the formal loop space ${\cal L}(\mathbb{R}^n):=C^{\infty}(S^1, \mathbb{R}^n)$, together with some important operations, following essentially \cite{DZ}. 
 From now on, every time two indices are repeated in a formula, they are summed over a suitable range, which is usually clear from the context. This applies also to the indices that denote derivatives; for instance in $\alpha^{(t)}_i\delta u^i_{(t)}$ sum over $i=1, \dots, n$ and sum over $t=0,1,\dots$ is intended.

Let $U\subset \mathbb{R}^n$ be an open subset of $\mathbb{R}^n$ with coordinates $u^1, \dots, u^n$.
Denote ${\cal A} = {\cal A}(U)$ the space of polynomials in the independent
variables $u^{i}_{(s)}$, $i=1, \dots, n$, $s=1, 2, \dots$, where $u^i_{(s)}$ has to be thought as the $s$-th derivative of $u^i$ with respect to the angular coordinate $x$. In the sequel where $u^i_{(0)}$ will appear, it will be identified with $u^i$. 
An element $f\in {\cal A}$ can be described as
\beq\label{dif.pol}
f(x; u; u_x, u_{xx}, \dots):=\sum_{m\geq 0} f_{i_1 s_1; \dots ;i_m s_m} (x;u) u^{i_1}_{(s_1)} \dots u^{i_m}_{(s_m)}
\eeq
where the coefficients $f_{i_1 s_1; \dots ; i_ms_m}(x;u)$ are smooth functions
on $S^1\times M$. These elements are usually called  differential
polynomials; observe that they are not required to be polynomial with respect to $(u^1, \dots, u^n)$ in general. In our case, the coefficients $f_{i_1 s_1; \dots ; i_ms_m}$ will not depend explicitly on the  $x$-coordinate. 
We recall the definition of the total derivative with respect to $x$ acting on $f\in {\cal A}$:
\beq\label{totalderivative}\p_x f=\frac{\p f}{\p x} +\frac {\p f}{\p u^i} u^{i}_x + \dots+
\frac{\p f}{\p u^i_{(s)} }u^{i}_{(s+1)}+\dots,\eeq
and analogously we denote with $f_{(s)}$ the expression $\p^s_x f$, $s=1, 2, \dots.$

We start defining the $0$-forms on ${\cal L}(\mathbb{R}^n)$, namely functions, that will be represented in terms of functionals. 

Consider the space ${\cal A}_{0,0} :={\cal A}/ \mathbb{R},$ obtained by identifying two differential polynomials that differ by a constants, and define ${\cal A}_{0,1}:= {\cal A}_{0,0}\, dx$. The differential operator $d: {\cal A}_{0,0} \to {\cal A}_{0,1}$, $f\mapsto df:=(\p_x f)\, dx$ allows us to consider the quotient vector space
\beq\label{functionals}
\Lambda_0 :={\cal A}_{0,1}/\mathrm{Im}(d),
\eeq
whose elements can be written as integrals
over the circle $S^1$
\beq\label{functional}
\bar I_{f}[u]:=\int_{S^1} f(x; u; u_x, u_{xx}, \dots) dx, 
\eeq
since if $f(x; u; u_x, u_{xx}, \dots) dx=(\p_x g)\, dx$ for some $g\in {\cal A}$, then the corresponding $I_f$ is the zero functional. 
Also let us remark that due to the definition of ${\cal A}_{0,0}$, functionals $I_{f}$ whose density $f$ is a constant are identified with the zero functional. Since we have mod out constants in the definition of ${\cal A}_{0,0}$ we obtain a short exact sequence 
$$0\to {\cal A}_{0,0} \stackrel{d}\to {\cal A}_{0,1} \stackrel{\pi}\to
\Lambda_0 \to 0,$$
where $\pi$ is the projection map. 


In order to define differential forms, we proceed as follows. First we introduce the Grassmann algebra $A_{\bullet, \bullet}$ over $\mathbb{R}$ generated by $\delta u^i_{(s)}$, where $i=1, \dots, n$, $s=0, 1, \dots$, ($\delta u^i_{(0)}$ is identified with $\delta u^i$) and by $dx$. This turns out to be bi-graded $A_{\bullet, \bullet}=\bigoplus_{k\geq 0,\, l=0,1}A_{k,l}$, where a monomial element $a_{k,0}\in A_{k,0}$ is of the form $a\delta u^{i_1}_{(s_1)}\wedge \dots \wedge\delta u^{i_k}_{(s_k)}$, $a\in \mathbb{R}$, while a monomial element $a_{k,1}\in A_{k,1}$ is written as $a\delta u^{i_1}_{(s_1)}\wedge \dots \wedge\delta u^{i_k}_{(s_k)}\wedge dx$, $a\in \mathbb{R}$. 
An arbitrary element $a_{k,0}\in A_{k,0}$ will be a {\em finite sum} of the form $a_{i_1, s_1, \dots, i_k, s_k} \delta u^{i_1}_{(s_1)}\wedge \dots \wedge \delta u^{i_k}_{(s_k)}$, where sum over repeated indices is assumed and where the coefficients $a_{i_1, s_1, \dots, i_k, s_k}$ are constants that are skew-symmetric for all exchanges of the pairs $(i_p, s_p)$ with $(i_q, s_q)$.

For $k\geq 1$, $l=0,1$ we define ${\cal A}_{k,l}:={\cal A}\otimes_{\mathbb{R}}A_{k,l}$, while ${\cal A}_{0,l}$ have been defined above. Notice that in the definition of ${\cal A}_{k,l}$, for $k\geq 1$, the ring ${\cal A}$ appears directly, without having to take quotient with respect to constants. 

Therefore also $\mathcal{A}_{\bullet, \bullet}$ is bi-graded, $\mathcal{A}_{\bullet, \bullet}=\bigoplus_{k\geq 0, \, l=0,1}\mathcal{A}_{k,l}$. 
Concretely, a homogeneous element $\omega$ of degree $k\geq 1$,  $\omega\in{\cal A}_{k,0}$ can be written as a finite sum
\beq\label{omega}
\omega= \frac{1}{k!} \omega_{i_1s_1; \dots; i_ks_k} \delta u^{i_1}_{(s_1)}\wedge
\dots \wedge \delta u^{i_k}_{(s_k)},
\eeq 
where sum over all repeated indices is intended and where the coefficients $\omega_{i_1s_1; \dots; i_ks_k}$ in  $\mathcal{A}$ are skew-symmetric for all exchanges of pairs $(i_p, s_p)$ with $(i_q, s_q)$. Analogously, a homogeneous element $\omega$ of bi-degree $(k,1)$ in $\mathcal{A}_{k,1}$ can be written as a finite sum
\beq\label{omega2}
\omega= \frac{1}{k!} \omega_{i_1s_1; \dots; i_ks_k} \delta u^{i_1}_{(s_1)}\wedge
\dots \wedge \delta u^{i_k}_{(s_k)}\wedge dx, 
\eeq
with the same conventions over repeated indices and skew-symmetry. 
So far we have not introduced the exterior differential ${\cal D}$ for the Grassmann algebra ${\cal A}_{\bullet, \bullet}$. We are going to construct it as a sum of two differentials. 
First we extend the differential $d: \mathcal{A}_{0,0}\to \mathcal{A}_{0,1}$ to a differential $d:\mathcal{A}_{k,0}\to \mathcal{A}_{k,1}$ defined 
as $d(\omega):=dx\wedge(\p_x \omega),$  where the derivation $\p_x$ satisfies the Leibniz rule
$$
\p_x (\omega_1 \wedge \omega_2) = \p_x \omega_1 \wedge \omega_2 + \omega_1
\wedge \p_x \omega_2
$$
 and the action of $\p_x$ is given by (\ref{totalderivative}) on the coefficients of the differential form and
by $\p_x \delta u^{i}_{(s)} = \delta u^{i}_{(s+1)}$ on the $\delta$-differentials.  It is clear that $d^2=0$ since $dx\wedge dx=0$. 
Similar to what happens for ${\cal A}_{0,0}$ and ${\cal A}_{0,1}$, for any $k\geq 1$ we have short exact sequences
$$0\to {\cal A}_{k,0} \stackrel{d}\to {\cal A}_{k,1} \stackrel{\pi}\to
\Lambda_k \to 0,$$
where $\pi$ is the projection onto the quotient $\Lambda_k:= {\cal A}_{k,1}/\mathrm{Im}(d)$. 
We call the elements of $\Lambda_k$ (local) $k$-forms on the loop space ${\cal L}(\mathbb{R}^n)$. 
We represent $k$-forms as integrals $$
\int dx\wedge\omega, \quad \omega \in {\cal A}_{k,0},
$$
and if $dx\wedge\omega=d(\alpha)$, for some $\alpha\in {\cal A}_{k,0}$, then the corresponding $k$-form in $\Lambda_k$ is identically zero.
As an example, notice that since any one-form can be written as 
$$\int dx\wedge\omega_{i, s}\,\delta u^i_{(s)}=\int dx\wedge\omega_{i, s}\, \p^s_x(\delta u^i),$$ 
using integration by parts we get
$$\int dx\wedge(-1)^s \p^s_x(\omega_{i,s})\delta u^i=\int dx\wedge \alpha_i \,\delta u^i ,$$
where $\alpha_i:=(-1)^s \p^s_x(\omega_{i,s}).$


We define now another differential $\delta: {\cal A}_{k,0}\to {\cal A}_{k+1,0}$. This will be extended to a differential $\delta : \Lambda_k \to \Lambda_{k+1}$ that will act as the exterior differential on forms on the loop space, so it will increase the degree of a form. 
This is the operator $\delta: {\cal A}_{k, 0}\to {\cal A}_{k+1, 0}$ which is defined on monomials and then extended by linearity as follows
\beq\label{deltaoperator}\delta\left(f\delta u^{i_1}_{(s_1)}\wedge\dots\wedge \delta u^{i_k}_{(s_k)}\right)=\f{\p f}{\p u^j_{(t)}}\delta u^j_{(t)}\wedge \delta u^{i_1}_{(s_1)}\wedge\dots\wedge \delta u^{i_k}_{(s_k)},\eeq
where again sum over repeated indices is intended. 
Notice that, as in the usual case, $\delta^2=0$ identically, as it can be checked immediately on monomials. 
The map $\delta$ on ${\cal A}_{k,1}$ is defined by the same formula with the requirement that $\delta dx =0$. 
\begin{lemma}\label{anticommutativity}
The operators $d$ and $\delta$ satisfies $\delta \circ d=-d \circ \delta$. 
\end{lemma}
\proof
It is easy to check that \beq\label{auxiliary}\f{\p }{\p u^i_{(s)}}\circ \p_x=\p_x \circ \f{\p }{\p u^i_{(s)}}+\f{\p }{\p u^i_{(s-1)}}, \; s\geq 1; \quad \f{\p }{\p u^i}\circ \p_x=\p_x \circ \f{\p }{\p u^i}.\eeq
If we consider a monomial 
$$\omega:=f\delta u^{i_1}_{(s_1)}\wedge \dots \wedge\delta u^{i_k}_{(s_k)} \in {\cal A}_{k,0},$$
then $$d\omega=(\p_x f)\,dx\wedge \delta u^{i_1}_{(s_1)}\wedge \dots \wedge\delta u^{i_k}_{(s_k)}+f\,dx\wedge \delta u^{i_1}_{(s_1+1)}\wedge \dots \wedge\delta u^{i_k}_{(s_k)}$$
$$+\dots+f\,dx\wedge \delta u^{i_1}_{(s_1)}\wedge \dots\wedge \delta u^{i_k}_{(s_k+1)},$$
so that 
$$\delta d\omega=\f{\p}{\p u^j_{(t)}}(\p_x f)\delta u^j_{(t)}\wedge dx \wedge \delta u^{i_1}_{(s_1)}\wedge \dots \wedge\delta u^{i_k}_{(s_k)}$$
$$+\f{\p f}{\p u^j_{(t)}}\delta u^j_{(t)}\wedge dx \wedge \delta u^{i_1}_{(s_1+1)}\wedge \dots \wedge\delta u^{i_k}_{(s_k)}+\dots$$
$$+\f{\p f}{\p u^j_{(t)}}\delta u^j_{(t)}\wedge dx \wedge \delta u^{i_1}_{(s_1)}\wedge \dots \wedge \delta u^{i_k}_{(s_k+1)}.$$
On the other hand $$d\delta \omega=\p_x\left( \f{\p f}{\p u^j_{(t)}}\right)dx\wedge\delta u^j_{(t)}\wedge \delta u^{i_1}_{(s_1)}\wedge \dots \wedge\delta u^{i_k}_{(s_k)}$$
$$+ \f{\p f}{\p u^j_{(t)}}dx\wedge\delta u^j_{(t+1)}\wedge \delta u^{i_1}_{(s_1)}\wedge \dots \wedge\delta u^{i_k}_{(s_k)}+\f{\p f}{\p u^j_{(t)}}dx\wedge\delta u^j_{(t)}\wedge \delta u^{i_1}_{(s_1+1)}\wedge \dots \wedge\delta u^{i_k}_{(s_k)}$$
$$+\dots+\f{\p f}{\p u^j_{(t)}}dx\wedge\delta u^j_{(t)}\wedge \delta u^{i_1}_{(s_1)}\wedge \dots\wedge \delta u^{i_k}_{(s_k+1)}.$$
Using relations \eqref{auxiliary} and $dx\wedge \delta u^j_{(t)}=-\delta u^j_{(t)}\wedge dx$ we get the result. 
\endproof

Lemma \ref{anticommutativity} allows to define $\delta : \Lambda_k \to \Lambda_{k+1}$. Indeed, let $\alpha\in \Lambda_k$ and let $\tilde \alpha$ and $\alpha'$ two of its representatives in ${\cal A}_{k,1}$. Since $\tilde\alpha-\alpha'=d\beta$, for some $\beta\in {\cal A}_{k,0}$, we have $\delta(\tilde\alpha)-\delta(\alpha')=\delta\circ d(\beta)=-d\circ\delta(\beta)$, which shows that $\delta(\tilde\alpha), \delta(\alpha')$ in ${\cal A}_{k+1,1}$ define the same form in $\Lambda_{k+1}$, call it $\delta(\alpha).$ 

The operator $\delta$ satisfies $\delta\circ \delta=0$ as it is immediate to see. 
The exterior differential of the differential Grassmann algebra ${\cal A}_{\bullet, \bullet}$ is the differential operator ${\cal D}:=\delta+d$, which indeed satisfies ${\cal D}\circ {\cal D}=0$. 
Moreover, each of the short exact sequences

$$
0\to {\cal A}_{k,0} \stackrel{d}\to {\cal A}_{k,1} \stackrel{\pi}\to
\Lambda_k \to 0
$$
 is included as a row in a bicomplex, called variational bicomplex, in which both arrows and columns
 are exact. The rows are related via the differential $\delta$. For more information see \cite{DT} and \cite{DZ}. 
%

 On $\Lambda_0$ the differential $\delta$ acts as follows
 $$\delta \int f\, dx=\int dx\wedge \f{\p f}{\p u^j_{(t)}}\delta u^j_{(t)},$$
 (sum over $t$ and $j$)
 and due to the fact that $\delta u^j_{(t)}=\p_x^t \delta u^j$, integrating by parts one obtains
\beq\label{euler-lagrange}
\delta \int f\, dx =\int dx \wedge\left((-1)^t \p_x^t \frac {\pal f} {\pal
u^{j}_{(t)}} \right)\delta u^j, 
\eeq
(sum over $t$ and $j$). We will use the notation
\beq\label{var.der}
\frac{\delta \bar f}{\delta u^i(x)} := (-1)^t \p_x^t \frac {\p f} {\p
u^{i}_{(t)}} 
\eeq
(sum over $t$) for the components of the 1-form, $\delta\bar f = \delta\int f\, dx$. This is the expression of variational derivative that appears in formula \eqref{Hform} in the Introduction. 
Sometimes, with abuse of notation we will denote with $\f{\delta f}{\delta u^i(x)}$  the expression 
$(-1)^t \p_x^t \frac {\p f} {\p
u^{i}_{(t)}}$ for a density $f$.

Let us also recall the following well-known result: a necessary and sufficient condition for 
$$
\frac{\delta\bar f}{\delta u^i(x)}=0, ~~i=1, \dots, n.
$$
is the existence of a differential polynomial $g=g(x;u;u_x;\dots)$ such that
$f=\pal_x g$.

\subsection{The space $\Lambda^1$ of vector fields on the formal loop space}
We denote with $\Lambda^1$ the space of vector fields on the formal loop space.
A vector field $\xi$ on ${\cal L}(\mathbb{R}^n)$ is a formal infinite sum of the form 
\beq\label{vectorfield}
\xi = \xi^0 \f{\p}{\p x} + 
\xi^{i,k} 
\f{\p}{\p u^{i}_{(k)}}, \quad \xi^0, \,\xi^{i,k} \in {\cal A},
\eeq
(sum over $i=1,\dots, n$, $k\geq 0$)
where $
\f{\p }{\p u^{i}_{(0)}}$ is $\f{\p}{\p u^i}.$
The derivative of a functional $\bar f=\int f(x; u; u_x, \dots) dx \in \Lambda_0$
along $\xi$ is given by
\beq\label{liederivativefunctional}
\xi \,\bar f := \int \left( \xi^0 \f{\p f}{\p x} +  \xi^{i,k} \f{\p f}{\p u^{i}_{(k)}}\right) dx,
\eeq
(sum over $i$ and $k$),
which is again an element in $\Lambda_0$. 
The Lie bracket of two vector fields $\xi$ and $\eta$ with components as given in \eqref{vectorfield} is a vector field  defined by
\eqa
&&[\xi,\eta] =(\xi^0 \eta^0_x - \eta^0 \xi^0_x
+\xi^{j,t}\f{\p \eta^0}{\p u^{j}_{(t)}}-
\eta^{j,t}\f{\p\xi^0}{\p u^{j}_{(t)}}) \f {\p}{\p x}
\nn\\
&&\quad
+\left( \xi^0 \f{\p \eta^{i,s}}{\p x} 
- \eta^0 \f{\p\xi^{i,s}}{\p x}
+ \xi^{j,t}\f{\p \eta^{i,s}} {\p u^{j}_{(t)}}
-\eta^{j,t} \f{\p\xi^{i,s}}{\pal u^{j}_{(t)}}\right) \f{\p}{\p
u^{i}_{(s)}},\label{commutator}
\eeqa
sum over $(i,j=1, \dots, n)$ and $s, t\geq 0.$

In the sequel we are going to deal with restricted classes of vector fields. First we have the following
\begin{definition}[\cite{O}, page 291]\label{evolutionary}
A vector field $$\xi=\xi^0 \f{\p}{\p x} + 
\xi^{i,k} 
\f{\p}{\p u^{i}_{(k)}}$$
is called {\em evolutionary} if $\xi^0=0$,  $\xi^{i,k}=\p_x^k(\xi^i)$ for $\xi^i \in {\cal A}$ and the differential functions $\xi^i$ do not depend explicitly on $x$. 
\end{definition}

An evolutionary vector field $\xi$ is parametrized by $n$ functions $\xi^1, \dots, \xi^n$  and can therefore be written as:
$$\xi=\p_x^k(\xi^i)\f{\p }{\p u^i_{(k)}}.$$
  In the case of an evolutionary vector field, the corresponding system of evolutionary PDEs is described via the system: 
 \beq\label{evolutionarypde}u^i_t=\xi^i(u, u_x, u_{xx}, \dots), \; i=1, \dots, n, \; \xi^i\in {\cal A}.\eeq
 
\subsection{Operations on forms and vector fields}
Let $\xi$ be an evolutionary vector field and $\omega\in {\cal A}_{k,0}$ as in \eqref{omega}, $k\geq 1$.
\begin{definition}
The contraction of $\omega\in {\cal A}_{k,0}$ with $\xi$ evolutionary vector field is given by the natural extension of the usual formula. Assuming the coefficients of the $k$-form 
$$\omega=\f{1}{k!}\omega_{i_1 s_1;\dots; i_{k}s_{k}}\delta u^{i_1}_{(s_1)}\wedge \dots \wedge \delta u^{i_{k}}_{(s_{k})},$$
antisymmetric w.r.t permutations of pairs $i_p,s_p\leftrightarrow i_q,s_q$, one obtains the following expression:
$$i_{\xi}\omega=\f{1}{(k-1)!}\p_x^k(\xi^j)\omega_{jk;i_1 s_1;\dots; i_{k-1}s_{k-1}}\delta u^{i_1}_{(s_1)}\wedge \dots \wedge \delta u^{i_{k-1}}_{(s_{k-1})},$$ 
An analogous formula holds for $\omega\in {\cal A}_{k,1}$. 
\end{definition}
It turns out that for $\xi$ evolutionary $i_{\xi}\circ d+d\circ i_{\xi}=0$ identically, so that contraction with respect to $\xi$ is a well-defined operation $i_{\xi}: \Lambda_k\to \Lambda_{k-1}$, see \cite{DZ}. 

Consider the functional $\bar f[u]:=\int_{S^1} f \; dx\; \in \Lambda_0$. Then the $1$ form $\omega:=\delta \bar f$ is given by 
$$\omega=\int dx\wedge \f{\p f}{\p u^i_{(t)}}\delta u^i_{(t)}=\int dx\wedge (-1)^t \p_x^t\left(\f{\p f}{\p u^i_{(t)}} \right)\delta u^{i}=\int dx\wedge \f{\delta f}{\delta u^i(x)}\delta u^i\;  \in \Lambda_1.$$
If $\xi$ is an evolutionary vector field, with components $\xi^1, \dots, \xi^n$, then 
$$i_{\xi}\omega=\int dx \;\p_x^t(\xi^i) \f{\p f}{\p u^i_{(t)}}=\int dx\; \xi^i(-1)^t \p_x^t\left( \f{\p f}{\p u^i_{(t)}} \right)=\int dx\; \xi^i\f{\delta f}{\delta u^i(x)}\; \in \Lambda_0.$$
Notice that $i_{\xi} \omega$ coincides (as it should be) with the Lie derivative of the functional $\bar f$ with respect to the vector field $\xi$ as given in the general formula \eqref{liederivativefunctional} (use the fact that $\xi$ is evolutionary and integrate by parts). 



Let us introduce also the Poisson structure as a Poisson bracket on functionals. 
Following \cite{DZ} we represent a Poisson structure
 in the form
\beq\label{pb}
\{ u^i(x), u^j(y)\} 
=\sum_s A^{ij}_s(u(x); u_x(x), u_{xx}(x),
\dots)\delta^{(s)}(x-y), 
\eeq
where $A^{ij}_s$ satisfies suitable conditions. In particular, in this paper we will focus on Poisson structures of the form  
$A^{ij}_1=g^{ij}(u^1, \dots, u^n)$ and $A^{ij}_0=\Gamma^{ij}_k u^k_x$, $A^{ij}_s=0$ $s\geq 2$ and these define a Poisson structure iff $g^{ij}$ is a non-degenerate contravariant flat metric (not necessarily positive definite) and $\Gamma^{ij}_k=-g^{il}\Gamma^j_{lk}$, where $\Gamma^j_{lk}$ are the Christoffel symbols of the Levi-Civita connection associated to the inverse metric $g_{ij}$ (see \cite{DN}). From now on this will be assumed. 

In the case we are dealing with, the Poisson bracket of two local functionals
$\bar f= \int f(x; u; u_x, \dots) dx$ and  $\bar h = \int h(x; u; u_x, \dots)
dx$  is then given by 
\eqa\label{Poisson1}
&&\{ \bar f, \bar h\}
= \iint dx dy \frac{\delta \bar f}{\delta u^i(x)}
\{ u^i(x), u^j(y)\} \frac{\delta \bar h}{\delta u^j(y)}
\nn\\
&&\quad
=\int dx \frac{\delta \bar f}{\delta u^i(x)} \left(g^{ij}(u^1, \dots, u^n)\p_x +\Gamma^{ij}_k u^k_x\right) \left( \frac{\delta \bar h}{\delta u^j(x)}\right)
\in \Lambda_0.
\eeqa
and thus it is again a local functional. 
From the dynamical point of view, the important property of the local Poisson brackets is that the Hamiltonian
systems
\beq\label{ham.pde}
u^i_t 
=\{ u^i(x), \bar H\} = \left(g^{ij}(u^1, \dots, u^n)\p_x +\Gamma^{ij}_k u^k_x\right)
\frac{\delta \bar H}{\delta u^j(x)}
\eeq
with Hamiltonians like
$$
\bar H= \int H(u; u_x, \dots) dx
$$
are evolutionary PDEs (\ref{evolutionarypde}). 

It will be important for the next section to interpret \eqref{Poisson1} in a slightly different way. Indeed we can view \eqref{Poisson1} as the pairing between the $1$-form $\delta \bar f$ and the vector field $P(\delta \bar g)$ obtained via the action of the Poisson structure $P$ on the $1$-form $\delta \bar h$. So we set
\beq\label{Poisson2}
\int dx \frac{\delta \bar f}{\delta u^i(x)} \left(g^{ij}(u^1, \dots, u^n)\p_x +\Gamma^{ij}_k u^k_x\right) \left( \frac{\delta \bar h}{\delta u^j(x)}\right)=<\delta \bar f, P \delta \bar h>,
\eeq
so that the evolutionary vector field $P\delta \bar g$ is given by 
\beq \label{vectorfield2}
P\delta \bar h=\p^s_x\left[\left(g^{ij}\p_x+\Gamma^{ij}_k u^k_x\right)(-1)^l\p_x^l\left(\f{\p h}{\p u^j_{(l)}} \right)\right]\f{\p }{\p u^i_{(s)}}.
\eeq
When $P\delta \bar h$ is paired with $\delta \bar f=\int dx\wedge \f{\delta \bar f}{\delta u^j(x)}\delta u^j$ one obtains exactly formula \eqref{Poisson2}.

The goal of the next section is to extend the formalism to include a Poisson brackets on $1$-form and to show that certain evolutionary equation can be still written in the form \eqref{ham.pde}, although {\em there is no Hamiltonian functional available}. 
\section{Poisson brackets on $1$-forms}
First we extend formula \eqref{Poisson2} and \eqref{vectorfield2} to deal with $1$-forms that are not closed.
Consider a $1$-forms $\alpha$  given by 
$$\alpha=\int_{S^1}dx\wedge\alpha^{(t)}_i\delta u^i_{(t)}.$$
Using integration by parts, we can always reduce it to {\em standard form} or {\em reduced form}, in the sense that only the differential $\delta u^i$ are involved. 
Indeed, one gets
$$\alpha=\int_{S^1}dx\wedge\alpha_i\delta u^i, $$
where $$\alpha_i=(-1)^t \p_x^t \f{\p }{\p u^i_{(t)} }\alpha_i^{(t)}.$$
We will often write a $1$-form $\alpha$ also skipping the integral sign, directly as $\alpha=\alpha_i\delta u^i$, but thinking that operations like integration by parts do not change $\alpha$. 
Moreover, when performing computation with the bracket we are going to define, we will always consider $1$-forms in standard form; this requirement is motivated especially by the way in which the Poisson structure acts on exact $1$-forms, namely on differentials of functionals. 

The evolutionary vector field $P\beta$ is given by 
\beq\label{vectorfield3}P\beta=\p^s_x\left[\left(g^{ij}\p_x+\Gamma^{ij}_k u^k_x\right)\beta_j\right]\f{\p }{\p u^i_{(s)}},\eeq
and it is parametrized by the $n$ functions \beq\label{vectorfield3components}(P\beta)^i:=\left(g^{ij}\p_x+\Gamma^{ij}_k u^k_x\right)\beta_j,\eeq 
which can be viewed as its components. 
Then the extension of \eqref{Poisson2} to general $1$-forms is given by 
\beq\label{Poisson3}
<\alpha, P\beta>=i_{P\beta}\alpha=\int \alpha_i\left(g^{ij}\p_x+\Gamma^{ij}_k u^k_x\right)\beta_j\, dx\; \in \Lambda_0.
\eeq

Before defining the Poisson bracket, we need to introduce the Lie derivative of a $1$-form $\alpha$ along a vector field $X$. 
Suppose $\alpha=\alpha^{(t)}_i \delta u^i_{(t)}$ and $X=X^{k,s}\f{\p }{\p u^k_{(s)}}$. Any element of the loop space $u\in C^{\infty}(S^1, M)$, can be seen as a section $\sigma_{u}: S^1 \to S^1 \times M$ of the trivial bundle $\pi: S^1 \times M \to S^1$. 
We can think of $u^i_{(t)}$, for $i=1,\dots, n$, $t=0, 1, \dots$ as coordinates on the infinite jet, describing the infinite prolongation of the section $\sigma_u$. 
In this framework, we can thus define the Lie derivative via the usual formula in coordinates as
\beq\label{Liederivative1}
\mathrm{Lie}_X\alpha=(\mathrm{Lie}_X\alpha)_i^{(t)}\delta u^i_{(t)}=\left( X^{k,s}\f{\p }{\p u^k_{(s)}}\alpha_i^{(t)}+\alpha_k^{(s)}\f{\p }{\p u^i_{(t)}}X^{k,s}\right)\delta u^i_{(t)}.
\eeq
In formula \eqref{Liederivative1} we should have written the expression for the Lie derivative more correctly as 
$$\mathrm{Lie}_X\alpha=\int dx\wedge (\mathrm{Lie}_X\alpha)_i^{(t)}\delta u^i_{(t)}.$$
In the sequel, whenever the integral sign is omitted, it is intended the the components of a form, including $\delta u^i_{(t)}$, are defined up to a total derivative.

We specialize formula \eqref{Liederivative1} to the case in which $X=P\beta$, assuming that $1$-form $\beta$ is written in standard form, $P$ is written in flat coordinates and considering also $\alpha$ in standard form. 
In this case we have the following expression:
\begin{definition}
The Lie derivative of $\alpha$ with respect to $P\beta$ written in flat coordinates, where $\alpha$ and $\beta$ are in standard form, is given by 
\beq\label{Liederivative2}
\begin{split}
\mathrm{Lie}_{P\beta}\alpha=\p_x^s\left(\eta^{kl}\p_x \beta_l \right)\f{\p }{\p u^k_{(s)}}\alpha_i \delta u^i+\alpha_k \f{\p }{\p u^i_{(s)}}\left(\eta^{kl}\p_x \beta_l \right)\delta u^i_{(s)}\\
=\left\{\p_x^s\left(\eta^{kl}\p_x \beta_l \right)\f{\p }{\p u^k_{(s)}}\alpha_i+(-1)^t \p_x^s \left[ \alpha_k \f{\p }{\p u^i_{(s)}}\left(\eta^{kl}\p_x \beta_l \right)\right]\right\}\delta u^i.
\end{split}
\eeq
\end{definition}

The Lie derivative thus defined satisfies Cartan's formula:
\begin{proposition}\label{cartanth}
If the Lie derivative is defined as in \eqref{Liederivative2}, then 
\beq\label{cartanformula}
\mathrm{Lie}_{P\beta}\alpha=i_{P\beta}\delta\alpha+\delta i_{P\beta}\alpha=i_{P\beta}\delta\alpha+\delta<\alpha, P\beta>,
\eeq
where $<\alpha, P\beta>$ denotes the pairing of the $1$-form $\alpha$ with the vector field $P\beta$. 
\end{proposition}
\proof
The last equality in \eqref{cartanformula} is obvious because by definition $i_{P\beta}\alpha=<\alpha, P\beta>$. We prove \eqref{cartanformula} in flat coordinates and assuming $\alpha$ and $\beta$ are in standard form. 
We have
$$\delta i_{P\beta}\alpha=\delta \int \,dx \,(\alpha_k\eta^{kl}\p_x\beta_l)=\int\, dx\,\wedge \f{\p }{\p u^i_{(s)}}\left( \alpha_k\eta^{kl}\p_x\beta_l \right)\delta u^i_{(s)}$$
$$=\int\, dx\,\wedge (-1)^s \p_x^s\left\{  \f{\p }{\p u^i_{(s)}}\left( \alpha_k\eta^{kl}\p_x\beta_l \right)\right\}\delta u^i$$
\beq\label{cartan3}=\int\, dx\,\wedge (-1)^s \p_x^s\left\{ \left(\eta^{kl}\p_x\beta_l \right)\f{\p \alpha_k }{\p u^i_{(s)}}\right\}\delta u^i+\int dx\,\wedge (-1)^s \p_x^s\left\{\alpha_k  \f{\p }{\p u^i_{(s)}}(\eta^{kl}\p_x\beta_l)\right\}\delta u^i.\eeq
On the other hand, using $\delta \alpha=\int dx\wedge \f{\p \alpha_i}{\p u^m_{(t)}}\delta u^m_{(t)}\wedge \delta u^i$, we have 
$$i_{P\beta}\delta \alpha=\int dx\wedge \p_x^s\left(\eta^{kl}\p_x \beta_l \right)\f{\p \alpha_i}{\p u^m_{(t)}} \delta u^i \delta^m_k \delta^s_t-\int dx\wedge \p_x^s\left(\eta^{kl}\p_x \beta_l \right)\f{\p \alpha_i}{\p u^m_{(t)}}\delta u^m_{(t)}\delta^i_k \delta^s_{0},$$
where $\delta^m_k$ and so on are Kronecker's delta. 
In the last expression, considering the second term on the right hand side, setting $s=0$ and $i=k$, renaming $t$ to $s$ and $m$ to $i$ and integrating by parts we obtain 
\beq\label{cartan4}i_{P\beta}\delta \alpha=\int dx\wedge \p_x^s\left(\eta^{kl}\p_x \beta_l \right)\f{\p \alpha_i}{\p u^k_{(s)}} \delta u^i-\int dx\wedge (-1)^s \p_x^s\left\{\left(\eta^{kl}\p_x\beta_l \right)\f{\p \alpha_k}{\p u^i_{(s)}} \right\} \delta u^i.\eeq
Summing the expression of $\delta i_{P\beta}\alpha$ given in \eqref{cartan3} with the expression of $i_{P\beta}\delta \alpha$ in \eqref{cartan4} we obtain immediately the second line of \eqref{Liederivative2}. The claim is proved. 
\endproof

Following \cite{MM} and \cite{GD}, we introduce the following bracket between two $1$-forms $\alpha$ and $\beta$. 
\begin{definition}
The Poisson bracket between two $1$-forms $\alpha$ and $\beta$ is defined as
\beq\label{Poissononforms}
\{ \alpha, \beta\}:=\mathrm{Lie}_{P\beta}\alpha-\mathrm{Lie}_{P\alpha}\beta+\delta<\beta, P\alpha>.
\eeq
\end{definition}
In the next Section we will prove the following properties of \eqref{Poissononforms}:
\begin{enumerate}
\item If $\alpha$ and $\beta$ are exact $1$-forms, $\alpha=\delta \bar{f}$, $\beta=\delta \bar{g}$, where $\bar f$, $\bar g$ are local functionals, then 
$\{\alpha, \beta\}=\{\delta \bar f, \delta \bar g\}=\delta \{\bar f, \bar g\}$, where $\{\bar f, \bar g\}$ is the usual Poisson bracket among local functionals;
\item $\{\cdot,\cdot\}$ equips the space of $1$-forms $\Lambda_1$ with a Lie algebra structure;
\item the Poisson structure induces an (anti)-homomorphism of Lie algebras between $(\Lambda_1, \, \{\cdot, \cdot\})$ and the space of evolutionary vector fields equipped with the Lie bracket given by the Lie commutator. 
\end{enumerate}

Let us remark that the bracket defined in \eqref{Poissononforms} does not in general fulfill Leibniz rule, as it is typical of Poisson structures on evolutionary PDEs. Despite this fact, we keep calling it Poisson bracket on $1$-forms, as it is customary to do in the infinite dimensional set-up.

First we derive a coordinate expression for the bracket $\{\alpha, \beta\}$ both in flat coordinates and in general coordinates. 

\begin{proposition}
Let $\alpha=\int dx \wedge \alpha_i \delta u^i$ and $\beta=\int dx \wedge \beta_j \delta u^j$ be two reduced $1$-forms. 
Then in flat coordinates $\{\alpha, \beta\}$ is a reduced $1$-form $\{\alpha, \beta\}=\int dx\wedge \{\alpha, \beta\}_j \delta u^j$, 
where 
\beq\label{poissoncoordinatesf}
\{\alpha, \beta\}_j=\eta^{kl}\left[\left( \p_x^{s+1}\beta_l\right)\f{\p \alpha_j}{\p u^k_{(s)}}-\left(\p_x^{s+1}\alpha_l \right)\f{\p \beta_j}{\p u^k_{(s)}}\right].
\eeq
\end{proposition}
\proof
In the proof we remove the integral sign, just to simplify notation, with the understanding that total derivatives with respect to $x$ can be safely eliminated. 
Using formula \eqref{Poissononforms} and the expressions for the Lie derivatives one gets
$$\{\alpha, \beta\}=\p_x^s\left( \eta^{kl}\p_x \beta_l\right)\f{\p \alpha_i}{\p u^k_{(s)}}\delta u^i+\alpha_k\f{\p }{\p u^i_{(s)}}\left(\eta^{kl}\p_x \beta_l \right)\delta u^i_{(s)}$$
$$-\p_x^s\left( \eta^{kl}\p_x \alpha_l\right)\f{\p \beta_i}{\p u^k_{(s)}}\delta u^i-\beta_k\f{\p }{\p u^i_{(s)}}\left(\eta^{kl}\p_x \alpha_l \right)\delta u^i_{(s)}$$
$$+\f{\p \beta_l}{\p u^i_{(s)}}\eta^{lm}\p_x\alpha_m\delta u^i_{(s)}+\beta_l \eta^{lm}\f{\p }{\p u^i_{(s)}}\p_x \alpha_m \delta u^i_{(s)}.$$
Therefore we have 
$$\{\alpha, \beta\}=\eta^{kl}\left[\left( \p_x^{s+1}\beta_l\right)\f{\p \alpha_i}{\p u^k_{(s)}}-\left(\p_x^{s+1}\alpha_l \right)\f{\p \beta_i}{\p u^k_{(s)}}\right]\delta u^i+\text{ residual terms}.$$
We show that the residual terms constitutes a total derivative with respect to $x$ and this proves \eqref{poissoncoordinatesf}.
Indeed the non-trivial residual terms are given by 
$$\alpha_k\f{\p }{\p u^i_{(s)}}\left(\eta^{kl}\p_x \beta_l \right)\delta u^i_{(s)}+\f{\p \beta_l}{\p u^i_{(s)}}\eta^{lm}\p_x\alpha_m\delta u^i_{(s)},$$
since the other two cancel out after relabeling indices. 
Using $\f{\p }{\p u^i_{(s)}}\circ \p_x=\p_x\circ \f{\p }{\p u^i_{(s)}}+\f{\p }{\p u^i_{(s-1)}}$ and renaming indices these can be rewritten as 
$$\alpha_k \eta^{kl}\p_x\left(\f{\p \beta_l}{\p u^i_{(s)}}\right)\delta u^i_{(s)}+\alpha_k \eta^{kl}\f{\p \beta_l}{\p u^i_{(s)}}\delta u^i_{(s+1)}+(\p_x\alpha_l)\eta^{kl}\f{\p \beta_l}{\p u^i_{(s)}}\delta u^i_{(s)}=\p_x\left\{\alpha_k \eta^{kl}\f{\p \beta_l}{\p u^i_{(s)}}\delta u^i_{(s)} \right\}.$$
\endproof

The expression of the bracket in general coordinates is given in the following:
\begin{proposition}
Let $\alpha=\int dx \wedge \alpha_i \delta u^i$ and $\beta=\int dx \wedge \beta_j \delta u^j$ be two reduced $1$-forms. 
Then in arbitrary coordinates $\{\alpha, \beta\}$ is a reduced $1$-form $\{\alpha, \beta\}=\int dx\wedge \{\alpha, \beta\}_i\delta u^i$, 
where 
\beq\begin{split}\label{poissoncoordinatesa}
\{\alpha, \beta\}_i=\p_x^s\left(g^{kl}\p_x\beta_l+\Gamma^{kl}_mu^m_x\beta_l\right)\f{\p\alpha_i}{\p u^k_{(s)}}-\p_x^s\left(g^{kl}\p_x\alpha_l+\Gamma^{kl}_mu^m_x\alpha_l\right)\f{\p\beta_i}{\p u^k_{(s)}}\\
+\left(\alpha_k\p_x\beta_l-\beta_k\p_x\alpha_l\right)\Gamma^{lk}_i-\alpha_k\beta_l\left[\Gamma^{k}_{is}\Gamma^{sl}_m-\Gamma^l_{is}\Gamma^{sk}_m\right]u^m_x,
\end{split}\eeq
where $g^{kl}$ is the flat (contravariant)-metric in the chosen coordinates and $\Gamma^{lk}_m$ are the corresponding Christoffel symbols. 
Moreover, if $\alpha_i$ and $\beta_i$ are functions depending only on the coordinates $u^1, \dots, u^n$, but not on their derivatives, then 
\beq\begin{split}\label{poissoncoordinatesa2}
\{\alpha, \beta\}_i=\left(\nabla_m \beta_l g^{kl}\nabla_k \alpha_i-\nabla_m \alpha_l g^{kl}\nabla_k \beta_i \right)u^m_x,
\end{split}\eeq
where $\nabla_m$ indicates covariant derivative with respect to the vector field $\f{\p }{\p u^m}$. 
\end{proposition}
\proof
To derive the general formula \eqref{poissoncoordinatesa} we use again \eqref{Poissononforms}, where $P$ is expressed this time in general coordinates. 
Expanding all the terms in \eqref{Poissononforms}, and skipping the integral sign to simplify notation we find
\beq\nonumber
\begin{split}
\{\alpha, \beta\}=\p_x^s\left(g^{kl}\p_x\beta_l+\Gamma^{kl}_mu^m_x\beta_l\right)\f{\p\alpha_i}{\p u^k_{(s)}} \delta u^i+\alpha_k \f{\p }{\p u^i_{(t)}}\left(g^{kl}\p_x\beta_l+\Gamma^{kl}_n u^n_x\beta_l\right)\delta u^i_{(t)}\\
-\p_x^s\left(g^{kl}\p_x\alpha_l+\Gamma^{kl}_mu^m_x\alpha_l\right)\f{\p\beta_i}{\p u^k_{(s)}}\delta u^i-\beta_k \f{\p }{\p u^i_{(t)}}\left(g^{kl}\p_x\alpha_l+\Gamma^{kl}_n u^n_x\alpha_l\right)\delta u^i_{(t)}\\
+\f{\p \beta_l}{\p u^i_{(t)}}\left(g^{lm}\p_x \alpha_m+\Gamma^{lm}_n u^n_x\alpha_m \right)\delta u^i_{(t)}+\beta_l \f{\p }{\p u^i_{(t)}}\left(g^{lm}\p_x \alpha_m+\Gamma^{lm}_n u^n_x\alpha_m \right)\delta u^i_{(t)}.
\end{split}
\eeq
From this expression we obtain immediately 
\beq\nonumber
\begin{split}
\{\alpha, \beta\}=\left\{\p_x^s\left(g^{kl}\p_x\beta_l+\Gamma^{kl}_mu^m_x\beta_l\right)\f{\p\alpha_i}{\p u^k_{(s)}} -\p_x^s\left(g^{kl}\p_x\alpha_l+\Gamma^{kl}_mu^m_x\alpha_l\right)\f{\p\beta_i}{\p u^k_{(s)}}\right\}\delta u^i\\
+\alpha_k \f{\p }{\p u^i_{(t)}}\left(g^{kl}\p_x\beta_l+\Gamma^{kl}_n u^n_x\beta_l\right)\delta u^i_{(t)}+\f{\p \beta_l}{\p u^i_{(t)}}\left(g^{lm}\p_x \alpha_m+\Gamma^{lm}_n u^n_x\alpha_m \right)\delta u^i_{(t)}\\
+\beta_l \f{\p }{\p u^i_{(t)}}\left(g^{lm}\p_x \alpha_m+\Gamma^{lm}_n u^n_x\alpha_m \right)\delta u^i_{(t)}-\beta_k \f{\p }{\p u^i_{(t)}}\left(g^{kl}\p_x\alpha_l+\Gamma^{kl}_n u^n_x\alpha_l\right)\delta u^i_{(t)},
\end{split}
\eeq
from which we recognize that the first two terms in \eqref{poissoncoordinatesa} in the first line, while the third line vanishses, relabeling indices. It remains to prove that the two terms in the second line in the previous expression are equal (up to total derivatives with respect to $x$) to the last line in \eqref{poissoncoordinatesa}.
In the expression 
$$\alpha_k \f{\p }{\p u^i_{(t)}}\left(g^{kl}\p_x\beta_l+\Gamma^{kl}_n u^n_x\beta_l\right)\delta u^i_{(t)},$$
we split the sum over $t$ into three terms corresponding to $t=0$, $t=1$ and $t\geq 2$ and then we use the identity $\f{\p }{\p u^i_{(t)}}\circ \p_x=\p_x\circ \f{\p }{\p u^i_{(t)}}+\f{\p }{\p u^i_{(t-1)}}$ for $t\geq 1$ and $\f{\p }{\p u^i_{(0)}}\circ \p_x=\p_x\circ \f{\p }{\p u^i_{(0)}}$. 
In this way we obtain 
\beq\label{casino}
\begin{split}
\alpha_k \f{\p }{\p u^i_{(t)}}\left(g^{kl}\p_x\beta_l+\Gamma^{kl}_n u^n_x\beta_l\right)\delta u^i_{(t)}+\f{\p \beta_l}{\p u^i_{(t)}}\left(g^{lm}\p_x \alpha_m+\Gamma^{lm}_n u^n_x\alpha_m \right)\delta u^i_{(t)}=\\
\alpha_k \left(\f{\p g^{kl}}{\p u^i}\p_x\beta_l \right)\delta u^i +\underbrace{\alpha_k\left(g^{kl}\p_x\f{\p \beta_l}{\p u^i}\right)\delta u^i}_{(1)}+\alpha_k\left( \f{\p \Gamma^{kl}_n}{\p u^i}u^n_x \beta_l\right)\delta u^i+\underbrace{\alpha_k \left(\Gamma^{kl}_nu^n_x\f{\p \beta_l}{\p u^i}\right)\delta u^i}_{(5)}\\
+\underbrace{\alpha_k\left(g^{kl}\left(\p_x \f{\p \beta_l}{\p u^i_{(1)}}+\f{\p \beta_l}{\p u^i} \right) \right)\delta u^i_{(1)}}_{(2)}+\alpha_k \Gamma^{kl}_i \beta_l\delta u^i_{(1)}+\underbrace{\alpha_k \left(\Gamma^{kl}_n u^n_x\f{\p \beta_l}{\p u^i_{(1)}}\right)\delta u^i_{(1)}}_{(6)}\\
+\underbrace{\alpha_k \left(g^{kl}\left(\p_x \f{\p \beta_l}{\p u^i_{(t)}}+\f{\p \beta_l}{\p u^i_{(t-1)}}\right) \right)\delta u^i_{(t)}}_{(3)\, \text{ for } t\geq2}+
\underbrace{\alpha_k\left(\Gamma^{kl}_n u^n_x\f{\p \beta_l}{\p u^i_{(t)}}\right)\delta u^i_{(t)}}_{(7)\, \text{ for } t\geq2}\\
+\underbrace{\f{\p \beta_l}{\p u^i_{(t)}}\left(g^{lm}\p_x \alpha_m \right)\delta u^i_{(t)}}_{(4)\, \text{ for } t\geq 0}+
\underbrace{\f{\p \beta_l}{\p u^i_{(t)}}\left(\Gamma^{lm}_n u^n_x \alpha_m \right)\delta u^i_{(t)}}_{(8)\, \text{ for } t\geq 0}.
\end{split}
\eeq
Summing the terms labeled $(1), (2), (3), (4)$ in \eqref{casino} we get immediately, after suitable relabeling of some indices, 
$$(1)+(2)+(3)+(4)=g^{kl}\p_x\left(\f{\p \beta_l}{\p u^i_{(t)}}\alpha_k \delta u^i_{(t)}\right) \; \text{sum over }t\geq 0.$$
On the other hand, summing all the terms labeled $(5), (6), (7), (8)$ in \eqref{casino} and relabeling $m$ to $k$ in term $(8)$ we obtain 
$$(5)+(6)+(7)+(8)=\alpha_k \f{\p \beta_l}{\p u^i_{(t)}}\left(\Gamma^{kl}_n+\Gamma^{lk}_n \right)u^n_x \delta u^i_{(t)}\; \text{sum over }t\geq 0.$$
Now it is well-known  that  $(\Gamma^{kl}_n+\Gamma^{lk}_n)=\f{\p}{\p u^n} g^{kl}, $
so that $\left(\Gamma^{kl}_n+\Gamma^{lk}_n \right)u^n_x=\p_x g^{kl}$. Therefore, the sum of all labeled terms in \eqref{casino} is equal to 
$$\p_x\left(g^{kl}\f{\p \beta_l}{\p u^i_{(t)}}\alpha_k \delta u^i_{(t)} \right),$$
and therefore it can be safely discarded. 

Now it remains to deal with 
\beq\label{residual5}\alpha_k \left(\f{\p g^{kl}}{\p u^i}\p_x\beta_l \right)\delta u^i +\alpha_k\left( \f{\p \Gamma^{kl}_n}{\p u^i}u^n_x \beta_l\right)\delta u^i+\alpha_k \Gamma^{kl}_i \beta_l\delta u^i_{(1)}.\eeq
First we express $ \f{\p \Gamma^{kl}_n}{\p u^i}u^n_x$ as a total derivative with respect to $x$, exchanging the indices $i$ and $n$ in $\f{\p \Gamma^{kl}_n}{\p u^i}$, using the zero curvature condition. Indeed, the zero curvature condition reads:
\beq\label{zerocurvaturecontra}
g^{is}\left(\p_s \Gamma^{jk}_l-\p_l \Gamma^{jk}_s \right)-\Gamma^{ij}_s \Gamma^{sk}_l+\Gamma^{ik}_s\Gamma^{sj}_l=0, 
\eeq
where $\p_s:=\f{\p }{\p u^s}$. From \eqref{zerocurvaturecontra}, lowering and renaming indices and multiplying by $u^n_x$ we obtain the identity 
\beq\label{zerocurvature2}
\p_i \Gamma^{kl}_nu^n_x=\p_n \Gamma^{kl}_i u^n_x+g_{mi}\Gamma^{mk}_s\Gamma^{sl}_n u^n_x-g_{mi}\Gamma^{ml}_s \Gamma^{sk}_n u^n_x.
\eeq
Substituting \eqref{zerocurvature2} in \eqref{residual5} and using $\p_n \Gamma^{kl}_i u^n_x=\p_x \Gamma^{kl}_i $, we obtain 
\beq\label{residual6}
\alpha_k \left(\f{\p g^{kl}}{\p u^i}\p_x \beta_l\right)\delta u^i+\alpha_k\beta_l(\p_x\Gamma^{kl}_i)\delta u^i+g_{mi}\alpha_k \beta_l\left[\Gamma^{mk}_s \Gamma^{sl}_n -\Gamma^{ml}_s\Gamma^{sk}_n\right]u^n_x\delta u^i+\alpha_k \beta_l\Gamma^{kl}_i \delta u^i_{(1)}.
\eeq
Recalling that $g_{mi}\Gamma^{mk}_s=-\Gamma^{k}_{is}$ we see that 
$$g_{mi}\alpha_k \beta_l\left[\Gamma^{mk}_s \Gamma^{sl}_n -\Gamma^{ml}_s\Gamma^{sk}_n\right]u^n_x\delta u^i=-\alpha_k\beta_l\left[\Gamma^{k}_{is}\Gamma^{sl}_m-\Gamma^l_{is}\Gamma^{sk}_m\right]u^m_x\delta u^i,$$
which appears as the last term in \eqref{poissoncoordinatesa}.
Integrating by parts $\alpha_k \beta_l\Gamma^{kl}_i \delta u^i_{(1)}$ and using $\p_i g^{kl}=\Gamma^{kl}_i+\Gamma^{lk}_i$, the remaining terms in \eqref{residual6} become (up to total derivatives)
$$\left[\alpha_k \Gamma^{kl}_i\p_x\beta_l+\alpha_k \Gamma^{lk}_i \p_x\beta_l-\p_x(\alpha_k \beta_l)\Gamma^{kl}_i\right]\delta u^i,$$
which after renaming indices is equal to 
$$(\alpha_k\p_x\beta_l-\beta_k\p_x\alpha_l )\Gamma^{lk}_i\delta u^i.$$
This is the third term in \eqref{poissoncoordinatesa}. This concludes the proof of \eqref{poissoncoordinatesa}.

Finally to prove \eqref{poissoncoordinatesa2} we use \eqref{poissoncoordinatesa} specializing it to the case in which $\alpha_i$ and $\beta_i$ depend only on the coordinates $u^1, \dots, u^n$. In this case then \eqref{poissoncoordinatesa} gives
\beq\nonumber\begin{split}\{\alpha, \beta\}_i=\left( g^{kl}\p_m \beta_l+\Gamma^{kl}_m \beta_l\right) u^m_x \f{\p \alpha_i}{\p u^k}-\left(g^{kl}\p_m \alpha_l+\Gamma^{kl}_m \alpha_l\ \right)u^m_x\f{\p \beta_i}{\p u^k}\\
+\Gamma^{lk}_i\left( \alpha_k \p_m \beta_l-\beta_k \p_m \alpha_l\right) u^m_x-\alpha_n \beta_k\left[g_{sl}\Gamma^{lk}_i\Gamma^{sn}_m \right]u^m_x+\alpha_k\beta_n\left[g_{sl}\Gamma^{lk}_i\Gamma^{sn}_m \right]u^m_x\\
=g^{kl}\left( \nabla_m \beta_l\right)u^m_x\f{\p \alpha_i}{\p u^k}-g^{kl}\left(\nabla_m \alpha_l \right)u^m_x\f{\p \beta_i}{\p u^k}\\
+\Gamma^{lk}_i\alpha_k u^m_x\left[ \f{\p \beta_l}{\p u^m}-\Gamma^n_{lm}\beta_n\right]-\Gamma^{lk}_i \beta_k u^m_x\left[\f{\p \alpha_l}{\p u^m}-\Gamma^n_{lm}\alpha_n \right]\\
=(\nabla_m \beta_l)u^m_x\left[g^{kl}\f{\p \alpha_i}{\p u^k}+\Gamma^{lk}_i \alpha_k \right]-(\nabla_m \alpha_l)u^m_x\left[g^{kl}\f{\p \beta_i}{\p u^k}+\Gamma^{lk}_i \beta_k \right]\\
=\left\{(\nabla_m \beta_l)g^{kl}(\nabla_k \alpha_i)-(\nabla_m \alpha_l)g^{kl}(\nabla_k \beta_i)\right\}u^m_x.
\end{split}
\eeq
Formula \eqref{poissoncoordinatesa2} is proved. 
\endproof

\section{Properties of the bracket}
In this Section we show that the bracket previously defined enjoys the same properties of the Poisson bracket on functionals. 
Let $\alpha:=\alpha_i \delta u^i$ and $\beta:=\beta_j \delta u^j$ be two $1$-forms written in standard form. 
Then the Poisson bracket between $\alpha$ and $\beta$ is again a $1$-form which is written in flat coordinates and in standard form as follows 
\beq\label{truePoisson1}
\{\alpha, \beta\}:=\left(\eta^{kl}(\p_x^{s+1}\beta_l)\f{\p\alpha_i}{\p u^k_{(s)}}-\eta^{kl}(\p_x^{s+1}\alpha_l)\f{\p \beta_i}{\p u^k_{(s)}} \right)\delta u^i, 
\eeq

First we show that when the bracket on forms is evaluated on exact $1$-forms, then it is equal to the differential of the standard Poisson bracket between the corresponding functionals. Indeed we have the following: 

\begin{proposition}\label{usualpoisson}
If $\alpha$ and $\beta$ are exact $1$-forms, $\alpha=\delta \bar{f}$, $\beta=\delta \bar{g}$, where $\bar f$, $\bar g$ are local functionals, then 
$\{\alpha, \beta\}=\{\delta \bar f, \delta \bar g\}=\delta \{\bar f, \bar g\}$, where $\{\bar f, \bar g\}$ is the usual Poisson bracket among local functionals, while $\{\alpha, \beta\}$ is the bracket on $1$-forms defined in the previous Section. 
\end{proposition}
\proof
By \eqref{Poissononforms}, using the fact that Lie derivative satisfies Cartan's identity \eqref{cartanformula}, we obtain 
$$\{ \alpha, \beta\}=i_{P\beta}\delta \alpha+\delta i_{P\beta}\alpha-i_{P\alpha}\delta \beta-\delta i_{P\alpha}\beta+\delta i_{P\alpha}\beta.$$
Therefore, since $\delta\circ \delta=0$, if $\alpha$ and $\beta$ are exact we have 
$$\{\alpha, \beta\}= \delta i_{P\beta}\alpha=\delta<\alpha, P\beta>.$$
In particular, if $\alpha=\delta \bar f$, $\beta=\delta \bar g$, we obtain 
$$\{\delta \bar f, \delta \bar g\}=\delta< \delta \bar f, P \delta \bar g>.$$
It is immediate to check that $< \delta \bar f, P \delta \bar g>=\{\bar f, \bar g\},$ thus the Proposition is proved. 
\endproof

\begin{proposition}\label{trueJacobi1}
The bracket defined in \eqref{truePoisson1} satisfies Jacobi identity. 
\end{proposition}
\proof
We compute the $i$-th component of $J(\alpha, \beta, \gamma):=\{\alpha, \{\beta,\gamma\}\}+\{\beta, \{\gamma,\alpha\}\}+\{\gamma, \{\alpha,\beta\}\}$. 
We have 
\beq\label{trueJacobi2}\nonumber
\begin{split}
J(\alpha, \beta, \gamma)_i=\eta^{kl}\left(\p_x^{s+1}\{\beta, \gamma\}_l\right)\f{\p \alpha_i}{\p u^k_{(s)}}-\eta^{kl}\left(\p_x^{s+1}\alpha_l \right)\f{\p \{\beta, \gamma\}_i}{\p u^k_{(s)}}\\
+ \eta^{kl}\left(\p_x^{s+1}\{\gamma, \alpha\}_l\right)\f{\p \beta_i}{\p u^k_{(s)}}-\eta^{kl}\left(\p_x^{s+1}\beta_l \right)\f{\p \{\gamma, \alpha\}_i}{\p u^k_{(s)}}\\
+\eta^{kl}\left(\p_x^{s+1}\{\alpha, \beta\}_l\right)\f{\p \gamma_i}{\p u^k_{(s)}}-\eta^{kl}\left(\p_x^{s+1}\gamma_l \right)\f{\p \{\alpha, \beta\}_i}{\p u^k_{(s)}}.
\end{split}
\eeq
Substituting the Poisson brackets  appearing in the previous expression with their formulas in flat coordinates and further expanding $J(\alpha, \beta, \gamma)_i$ we obtain 
\beq\nonumber
\begin{split}
J(\alpha, \beta, \gamma)_i=\eta^{kl}\eta^{pq}\left\{ \p^{s+1}_x\left[\left(\p_x^{t+1}\gamma_q \right)\f{\p \beta_l}{\p u^p_{(t)}}-\left(\p_x^{t+1}\beta_q\right)\f{\p \gamma_l}{\p u^p_{(t)}} \right]\f{\p \alpha_i}{\p u^k_{(s)}} \right. \\
-\left(\p_x^{s+1}\alpha_l\right)\f{\p }{\p u^k_{(s)}}\left(\p_x^{t+1}\gamma_q \right)\f{\p \beta_i}{\p u^{p}_{(t)}} -\left(\p_x^{s+1}\alpha_l\right)\left(\p_x^{t+1}\gamma_q\right)\f{\p^2 \beta_i}{\p u^k_{(s)}\p u^p_{(t)}}\\
+(\p_x^{s+1}\alpha_l )\f{\p }{\p u^k_{(s)}}\left(\p_x^{t+1}\beta_q\right)\f{\p \gamma_i}{\p u^p_{(t)}}+\left(\p_x^{s+1}\alpha_l \right)\left(\p_x^{t+1}\beta_q\right)\f{\p^2 \gamma_i}{\p u^p_{(t)} \p u^k_{(s)}}\\
+\p^{s+1}_x\left[\left(\p_x^{t+1}\alpha_q \right)\f{\p \gamma_l}{\p u^p_{(t)}}-\left(\p_x^{t+1}\gamma_q\right)\f{\p \alpha_l}{\p u^p_{(t)}} \right]\f{\p \beta_i}{\p u^k_{(s)}}\\
-\left(\p_x^{s+1}\beta_l\right)\f{\p }{\p u^k_{(s)}}\left(\p_x^{t+1}\alpha_q \right)\f{\p \gamma_i}{\p u^{p}_{(t)}} -\left(\p_x^{s+1}\beta_l\right)\left(\p_x^{t+1}\alpha_q\right)\f{\p^2 \gamma_i}{\p u^k_{(s)}\p u^p_{(t)}}\\
+(\p_x^{s+1}\beta_l )\f{\p }{\p u^k_{(s)}}\left(\p_x^{t+1}\gamma_q\right)\f{\p \alpha_i}{\p u^p_{(t)}}+\left(\p_x^{s+1}\beta_l \right)\left(\p_x^{t+1}\gamma_q\right)\f{\p^2 \alpha_i}{\p u^p_{(t)} \p u^k_{(s)}} \\
+\p^{s+1}_x\left[\left(\p_x^{t+1}\beta_q \right)\f{\p \alpha_l}{\p u^p_{(t)}}-\left(\p_x^{t+1}\alpha_q\right)\f{\p \beta_l}{\p u^p_{(t)}} \right]\f{\p \gamma_i}{\p u^k_{(s)}}\\
-\left(\p_x^{s+1}\gamma_l\right)\f{\p }{\p u^k_{(s)}}\left(\p_x^{t+1}\beta_q \right)\f{\p \alpha_i}{\p u^{p}_{(t)}} -\left(\p_x^{s+1}\gamma_l\right)\left(\p_x^{t+1}\beta_q\right)\f{\p^2 \alpha_i}{\p u^k_{(s)}\p u^p_{(t)}}\\
\left.+(\p_x^{s+1}\gamma_l )\f{\p }{\p u^k_{(s)}}\left(\p_x^{t+1}\alpha_q\right)\f{\p \beta_i}{\p u^p_{(t)}}+\left(\p_x^{s+1}\gamma_l \right)\left(\p_x^{t+1}\alpha_q\right)\f{\p^2 \beta_i}{\p u^p_{(t)} \p u^k_{(s)}}\right\}.
\end{split}
\eeq
Let us focus our attention on the terms involving second derivatives of $\alpha_i$. 
We have 
\beq\label{term1}\eta^{kl}\eta^{pq}\left\{ \left(\p_x^{s+1}\beta_l\right)\left( \p_x^{t+1}\gamma_q\right)- \left(\p_x^{s+1}\gamma_l\right)\left( \p_x^{t+1}\beta_q\right)\right\}\f{\p^2 \alpha_i}{\p u^k_{(s)} \p u^p_{(t)}},\eeq
which can be written as $\eta^{kl}\eta^{pq}T_{ikplq}$, where 
$$T_{ikplq}:=\left\{ \left(\p_x^{s+1}\beta_l\right)\left( \p_x^{t+1}\gamma_q\right)- \left(\p_x^{s+1}\gamma_l\right)\left( \p_x^{t+1}\beta_q\right)\right\}\f{\p^2 \alpha_i}{\p u^k_{(s)} \p u^p_{(t)}}.$$
Now observe that $T_{ipklq}=T_{ikplq}=-T_{ipkql}$, so $T$ is symmetric under exchange of $p$ and $k$ and anti-symmetric under exchange of $q$ and $l$. 
Therefore 
$$\eta^{kl}\eta^{pq}T_{ipklq}=-\eta^{kl}\eta^{pq}T_{ipkql}=-\eta^{kl}\eta^{pq}T_{ikpql}=-\eta^{pq}\eta^{kl}T_{ipklq},$$
where in the last equality we have renamed the summed indices. Therefore, the expression \eqref{term1} is zero. Analogously for the terms one obtains collecting the second derivatives in $\beta_i$ and $\gamma_i$. 

It remains to deal with term containing the first derivatives in $\alpha_i$ (a completely analogous computation will show that also the term containing first derivatives of $\beta_i$ and $\gamma_i$ will indeed vanish and it will be skipped). 
The term we are interested in is given by:
\beq\label{term2}
\begin{split}
\eta^{kl}\eta^{pq}\left\{ \p^{s+1}_x\left[\left(\p_x^{t+1}\gamma_q \right)\f{\p \beta_l}{\p u^p_{(t)}}-\left(\p_x^{t+1}\beta_q\right)\f{\p \gamma_l}{\p u^p_{(t)}} \right]\f{\p \alpha_i}{\p u^k_{(s)}}\right. \\
\left.+(\p_x^{s+1}\beta_l )\f{\p }{\p u^k_{(s)}}\left(\p_x^{t+1}\gamma_q\right)\f{\p \alpha_i}{\p u^p_{(t)}}-\left(\p_x^{s+1}\gamma_l\right)\f{\p }{\p u^k_{(s)}}\left(\p_x^{t+1}\beta_q \right)\f{\p \alpha_i}{\p u^{p}_{(t)}} \right\}.
\end{split}
\eeq
In the last two terms, we rename summed indices in order to get 
\beq\label{term3}
\begin{split}
\eta^{kl}\eta^{pq}\f{\p \alpha_i}{\p u^k_{(s)}}\left\{ \p^{s+1}_x\left[\left(\p_x^{t+1}\gamma_q \right)\f{\p \beta_l}{\p u^p_{(t)}}-\left(\p_x^{t+1}\beta_q\right)\f{\p \gamma_l}{\p u^p_{(t)}} \right]\right. \\
\left.+(\p_x^{t+1}\beta_q )\f{\p }{\p u^p_{(t)}}\left(\p_x^{s+1}\gamma_l\right)-\left(\p_x^{t+1}\gamma_q\right)\f{\p }{\p u^p_{(t)}}\left(\p_x^{s+1}\beta_l \right) \right\}.
\end{split}
\eeq
It is immediate to see that \eqref{term3} vanishes identically if 
\beq\label{term4}
\begin{split}
 \p^{s+1}_x\left[\left(\p_x^{t+1}\gamma_q \right)\f{\p \beta_l}{\p u^p_{(t)}}\right]=\left(\p_x^{t+1}\gamma_q\right)\f{\p }{\p u^p_{(t)}}\left(\p_x^{s+1}\beta_l \right),\\
 \p^{s+1}_x\left[\left(\p_x^{t+1}\beta_q\right)\f{\p \gamma_l}{\p u^p_{(t)}} \right]=\left(\p_x^{t+1}\beta_q \right)\f{\p }{\p u^p_{(t)}}\left(\p_x^{s+1}\gamma_l\right).
\end{split}
\eeq
Obviously, it is sufficient to prove the first of \eqref{term4}. 
We expand the left hand side of the first of \eqref{term4} using the binomial formula as 
\beq\label{term5}
\begin{split}
\sum_{s\geq 0,\, t\geq0} \sum_{l=0}^{s+1}\binom{s+1}{l}\left( \p_x^{t+1+l}\gamma_q\right)\p_x^{s+1-l}\left( \f{\p \beta_l}{\p u^p_{(t)}}\right)\\
=\sum_{s\geq 0} \sum_{l=0}^{s+1}\binom{s+1}{l}\sum_{t\geq 0}\left( \p_x^{t+1+l}\gamma_q\right)\p_x^{s+1-l}\left( \f{\p \beta_l}{\p u^p_{(t)}}\right).
\end{split}
\eeq
To facilitate the comparison with other terms, we split the second line in \eqref{term5} as follows:
\beq\label{termauxiliary1}
\begin{split}
\sum_{s\geq 0} \sum_{l=0}^{s}\binom{s+1}{l}\sum_{t\geq 0}\left( \p_x^{t+1+l}\gamma_q\right)\p_x^{s+1-l}\left( \f{\p \beta_l}{\p u^p_{(t)}}\right)\\
+\sum_{s\geq 0} \sum_{t\geq 0}\left( \p_x^{t+1+s+1}\gamma_q\right)\left( \f{\p \beta_l}{\p u^p_{(t)}}\right).
\end{split}
\eeq
Now we rewrite \eqref{termauxiliary1} splitting its first term, separating the sum in $t$, in the following way:
\beq\label{termauxiliary2}
\begin{split}
\sum_{s\geq 0} \sum_{l=0}^{s}\binom{s+1}{l}\sum_{t\geq s+1-l}\left( \p_x^{t+1+l}\gamma_q\right)\p_x^{s+1-l}\left( \f{\p \beta_l}{\p u^p_{(t)}}\right)\\
+\sum_{s\geq 0} \sum_{l=0}^{s}\binom{s+1}{l}\sum_{0\leq t<s+1-l}\left( \p_x^{t+1+l}\gamma_q\right)\p_x^{s+1-l}\left( \f{\p \beta_l}{\p u^p_{(t)}}\right)\\
+\sum_{s\geq 0} \sum_{t\geq 0}\left( \p_x^{t+1+s+1}\gamma_q\right)\left( \f{\p \beta_l}{\p u^p_{(t)}}\right).
\end{split}
\eeq
Putting together the first and the last line in \eqref{termauxiliary2} and rewriting the second line we get that the left hand side of the first of \eqref{term4} is given by:
\beq\label{term5.2}
\begin{split}
\sum_{s\geq 0} \sum_{l=0}^{s+1}\binom{s+1}{l}\sum_{t\geq s+1-l}\left( \p_x^{t+1+l}\gamma_q\right)\p_x^{s+1-l}\left( \f{\p \beta_l}{\p u^p_{(t)}}\right)\\
+\sum_{s\geq 0} \sum_{l=0}^{s}\binom{s+1}{l}\sum_{0\leq t< s+1-l}\left( \p_x^{t+1+l}\gamma_q\right)\p_x^{s+1-l}\left( \f{\p \beta_l}{\p u^p_{(t)}}\right).
\end{split}
\eeq
 We recall the following identity which can be easily proved by induction 
\beq\label{identity}
\f{\p }{ \p u^p_{(t)}}\circ \p^n_x=\sum_{l=\max\{0,n-t\}}^{n} \binom{n}{l}\p_x^l \circ \f{\p }{\p u^p_{(t-n+l)}}.
\eeq
To expand the right hand side of the first of \eqref{term4} we use \eqref{identity} and get
\beq\label{term6}
\left(\p_x^{t+1}\gamma_q\right)\f{\p }{\p u^p_{(t)}}\left(\p_x^{s+1}\beta_l \right)=\sum_{s\geq 0, \, t\geq 0}\sum_{l=\max\{0, s+1-t\}}^{s+1}\binom{s+1}{l}\left( \p_x^{t+1}\gamma_q  \right)\p_x^l \left(\f{\p \beta_l}{\p u^p_{(t-s-1+l)}}\right).
\eeq
We split \eqref{term6} into two pieces, according if $t\geq s+1$ or $t\leq s$. 
We obtain 
\beq\label{term7}
\begin{split}
\left(\p_x^{t+1}\gamma_q\right)\f{\p }{\p u^p_{(t)}}\left(\p_x^{s+1}\beta_l \right)=\sum_{s\geq 0}\sum_{t\geq s+1} \sum_{l=0}^{s+1}\binom{s+1}{l}\left( \p_x^{t+1}\gamma_q  \right)\p_x^l \left(\f{\p \beta_l}{\p u^p_{(t-s-1+l)}}\right)\\
+\sum_{s\geq 0}\sum_{t=0}^s \sum_{l=s+1-t}^{s+1}\binom{s+1}{l}\left( \p_x^{t+1}\gamma_q  \right)\p_x^l \left(\f{\p \beta_l}{\p u^p_{(t-s-1+l)}}\right).
\end{split}
\eeq
Defining the new index $l':=s+1-l$ we obtain 
\beq\label{term8}
\begin{split}
\left(\p_x^{t+1}\gamma_q\right)\f{\p }{\p u^p_{(t)}}\left(\p_x^{s+1}\beta_l \right)=\sum_{s\geq 0}\sum_{t\geq s+1} \sum_{l'=0}^{s+1}\binom{s+1}{l'}\left( \p_x^{t+1}\gamma_q  \right)\p_x^{s+1-l'} \left(\f{\p \beta_l}{\p u^p_{(t-l')}}\right)\\
+\sum_{s\geq 0}\sum_{t=0}^s \sum_{l'=0}^{t}\binom{s+1}{l'}\left( \p_x^{t+1}\gamma_q  \right)\p_x^{s+1-l'} \left(\f{\p \beta_l}{\p u^p_{(t-l')}}\right).
\end{split}
\eeq
We can rewrite the first term of the right hand side of \eqref{term8} as 
$$\sum_{s\geq 0}\sum_{t\geq s+1} \sum_{l'=0}^{s+1}\binom{s+1}{l'}\left( \p_x^{t+1}\gamma_q  \right)\p_x^{s+1-l'} \left(\f{\p \beta_l}{\p u^p_{(t-l')}}\right)=$$
$$=\sum_{s\geq 0}\sum_{l'=0}^{s+1}\binom{s+1}{l'}\sum_{t\geq s+1}\left( \p_x^{t+1}\gamma_q  \right)\p_x^{s+1-l'} \left(\f{\p \beta_l}{\p u^p_{(t-l')}}\right),$$
and defining $t'=t-l'$ we obtain 
$$\sum_{s\geq 0}\sum_{l'=0}^{s+1}\binom{s+1}{l'}\sum_{t\geq s+1}\left( \p_x^{t+1}\gamma_q  \right)\p_x^{s+1-l'} \left(\f{\p \beta_l}{\p u^p_{(t-l')}}\right)=$$
$$=\sum_{s\geq 0}\sum_{l'=0}^{s+1}\binom{s+1}{l'}\sum_{t'\geq s+1-l'}\left( \p_x^{t'+1+l'}\gamma_q  \right)\p_x^{s+1-l'} \left(\f{\p \beta_l}{\p u^p_{(t')}}\right), $$
which we recognize as the first of the terms in \eqref{term5.2}.

Now we rewrite the second term of the right hand side of \eqref{term8} as follows:
$$\sum_{s\geq 0}\sum_{t=0}^s \sum_{l'=0}^{t}\binom{s+1}{l'}\left( \p_x^{t+1}\gamma_q  \right)\p_x^{s+1-l'} \left(\f{\p \beta_l}{\p u^p_{(t-l')}}\right)$$
$$=\sum_{s\geq 0}\sum_{l'=0}^s\sum_{t\geq l'}^s\binom{s+1}{l'}\left( \p_x^{t+1}\gamma_q  \right)\p_x^{s+1-l'} \left(\f{\p \beta_l}{\p u^p_{(t-l')}}\right).$$
Again defining $t'=t-l'$ we can rewrite the previous expression as 
$$\sum_{s\geq 0}\sum_{l'=0}^s\sum_{0\leq t'<s+1-l'}\binom{s+1}{l'}\left( \p_x^{t'+1+l'}\gamma_q  \right)\p_x^{s+1-l'} \left(\f{\p \beta_l}{\p u^p_{(t')}}\right),$$
which we recognize as the second of the terms in \eqref{term5.2}. This proves that the first of \eqref{term4} is indeed an identity. 

Therefore the term containing the first derivatives of $\alpha_i$ vanishes and similarly for the terms containing the first derivatives of $\beta_i$ and $\gamma_i$. 
The Jacobi identity is therefore proved. 
\endproof

\begin{corollary}\label{Liealgebraforms}
The bracket \eqref{truePoisson1} equips the vector space of $1$-forms $\Lambda_1$ with a Lie algebra structure. 
\end{corollary}
\proof 
Clearly $\{\alpha, \beta\}=-\{\beta, \alpha\}$ and $\{\cdot, \cdot\}$ is $\mathbb{R}$-bilinear. By Proposition \ref{trueJacobi1} it also fulfills Jacobi identity. 
\endproof

Let us recall that the vector space of evolutionary vector fields $\Lambda^1_{\mathrm{ev}}$ is naturally equipped with the Lie product given by the commutator: 
\beq\label{commutatorvector}
[\xi, \eta]^p=\xi^i_{(s)}\f{\p }{\p u^i_{(s)}}\eta^p-\eta^i_{(s)}\f{\p }{\p u^i_{(s)}}\xi^p,
\eeq 
where $\xi^i_{(s)}:=\p_x^s(\xi^i)$ and $\xi^i$ is the $i$-th component of the evolutionary vector field $\xi$. It is known that $(\Lambda^1_{\mathrm{ev}}, [\cdot, \cdot])$ is a Lie algebra.

\begin{proposition}\label{antih}
 The Poisson structure $P$ sending $1$-forms to evolutionary vector fields satisfies the identity \beq\label{trueantihomomorphism} P\{\alpha, \beta\}=-[P\alpha, P\beta].\eeq 
 Therefore $P$ is an (anti)-homomorphism of Lie algebras, $P: (\Lambda_1, \{ \cdot, \cdot\})\to (\Lambda^1_{\mathrm{ev}}, [\cdot, \cdot])$.
 \end{proposition}
\proof
We prove the claim in flat coordinates. We have:
$$(P\{\alpha, \beta\})^i=\eta^{ik}\p_x \{\alpha, \beta\}_k=\eta^{ik}\eta^{pl}\p_x\left((\p_x^{s+1}\beta_l)\f{\p\alpha_k}{\p u^p_{(s)}}-(\p_x^{s+1}\alpha_l)\f{\p \beta_k}{\p u^p_{(s)}} \right),$$
and further expanding
\beq\label{lefthandside}
\begin{split}
\eta^{ik}\eta^{pl}\sum_{s\geq 0}\left((\p_x^{s+2}\beta_l)\f{\p\alpha_k}{\p u^p_{(s)}}-(\p_x^{s+2}\alpha_l)\f{\p \beta_k}{\p u^p_{(s)}}\right)\\
+\eta^{ik}\eta^{pl}\sum_{s\geq 0}\left((\p_x^{s+1}\beta_l)\p_x\f{\p\alpha_k}{\p u^p_{(s)}}-(\p_x^{s+1}\alpha_l)\p_x\f{\p \beta_k}{\p u^p_{(s)}}\right),
\end{split}
\eeq
where we have explicitly inserted the summation symbol to make comparison with the next expression easier. 
On the other hand 
$$-[P\alpha, P\beta]^i=-\left((P\alpha)^p_{(s)}\f{\p }{\p u^p_{(s)}}(P\beta)^i-(P\beta)^p_{(s)}\f{\p }{\p u^p_{(s)}}(P\alpha)^i\right)$$
$$=-\left(\p_x^{s}\left( \eta^{pl}\p_x \alpha_l\right)\f{\p }{\p u^p_{(s)}}\left(\eta^{ik}\p_x \beta_k\right)-\p_x^{s}\left( \eta^{pl}\p_x \beta_l\right)\f{\p }{\p u^p_{(s)}}\left(\eta^{ik}\p_x \alpha_k\right)\right).$$
Using \eqref{identity}, this last expression is equal to 
$$-\eta^{ik}\eta^{pl}\sum_{s\geq 1}\left((\p_x^{s+1} \alpha_l)\left( \p_x\f{\p \beta_k}{\p u^p_{(s)}}+\f{\p \beta_k}{\p u^p_{(s-1)}}\right)-(\p_x^{s+1} \beta_l)\left(\p_x\f{\p \alpha_k}{\p u^p_{(s)}}+\f{\p \alpha_k}{\p u^p_{(s-1)}}\right)\right)$$
$$-\eta^{ik}\eta^{pl}\left((\p_x \alpha_l) \p_x\f{\p \beta_k}{\p u^p_{(0)}}-(\p_x \beta_l)\p_x\f{\p \alpha_k}{\p u^p_{(0)}}\right),$$
which we can rearrange 
as 
\beq\label{righthandside}
\begin{split}-\eta^{ik}\eta^{pl}\sum_{s\geq 0}\left((\p_x^{s+1} \alpha_l)\left( \p_x\f{\p \beta_k}{\p u^p_{(s)}}\right)-(\p_x^{s+1} \beta_l)\left(\p_x\f{\p \alpha_k}{\p u^p_{(s)}}\right)\right)\\
-\eta^{ik}\eta^{pl}\sum_{s\geq 0}\left((\p_x^{s+2} \alpha_l)\left( \f{\p \beta_k}{\p u^p_{(s)}}\right)-(\p_x^{s+2} \beta_l)\left(\f{\p \alpha_k}{\p u^p_{(s)}}\right)\right).
\end{split}
\eeq
Comparing \eqref{lefthandside} and \eqref{righthandside} we obtain $P\{\alpha, \beta\}=-[P\alpha, P\beta]$. 
\endproof

The following Proposition singles out the vector fields in $\Lambda^1_{\text{ev}}$ that are in the image of $P$. 
\begin{proposition}
Let $P:(\Lambda_1, \{ \cdot, \cdot\})\to (\Lambda^1_{\mathrm{ev}}, [\cdot, \cdot])$ be the Lie algebra (anti)-homomorphism given by the Poisson structure. Then the image of $P$ in $(\Lambda^1_{\mathrm{ev}}, [\cdot, \cdot])$ is a Lie subalgebra given by the evolutionary vector fields that are tangent to the symplectic leaves of $P$.
\end{proposition}
\proof
The fact that the image of $P$ in $(\Lambda^1_{\mathrm{ev}}, [\cdot, \cdot])$ is a Lie subalgebra is clear. Moreover, if $\alpha\in \Lambda_1$, $\alpha=\int dx\wedge \alpha_i \delta u^i$, then $P\alpha=\p_x^s\left(\eta^{kl}\p_x \alpha_l \right)\f{\p }{\p u^k_{(s)}}$ (in flat coordinates) which is an evolutionary vector field tangent to the symplectic leaves of $P$. To see this, it is sufficient to show that $(P\alpha)(f_i)=0$, where $f_i$ are the Casimirs of $P$, since the symplectic leaves are described by $\{f_1=c_1, \dots, f_n=c_n\}$, $c_1, \dots, c_n$ constants. It is well-known that the Casimirs of $P$ are given by the local functionals $f_i=\int u^i \, dx$ (here $u^1, \dots u^n$ are flat coordinates). Now the vector field $P\alpha$ applied to the functional $f_i$ is computed taking the pairing of $P\alpha$ with the exact $1$-form $\delta f_i$, $<P\alpha, \delta f_i>=\int dx\, 1\, \eta^{kl}\p_x \alpha_l$ which is clearly zero being a total derivative. 
Viceversa, we have to show that every evolutionary vector field $\xi$ tangent to the symplectic leaves of $P$ is of the form $P\alpha$ for some $\alpha$. Since $\xi$ is evolutionary, $\xi=\p_x^s \xi^j \f{\p }{\p u^j_{(s)}}$, for some functions $\xi^1, \dots, \xi^n$. Imposing that $\xi$ is tangent to the symplectic leaves of $P$, namely $<\xi, \delta f_i>=0$ for $i=1,\dots, n$ we obtain $\int dx\, \xi^j=0$ for all $j=1, \dots n$. Therefore, each $\xi^j$ is a total derivative $\xi^j=\p_x \zeta^j$. Since $\eta^{kl}$ is invertible, we can write $\zeta^j=\eta^{ji}\alpha_i$ for some $\alpha_i$. Therefore we obtain $\xi^j=\p_x (\eta^{ji}\alpha_i)=\eta^{ji}\p_x \alpha_i$ which proves the claim. 
\endproof

\section{$F$-manifolds and Poisson brackets on $1$-forms}
\begin{de}
An  $F$-manifold $M$ with compatible flat connection is a manifold endowed with a commutative associative product $\circ$ on vector fields and a symmetric flat connection $\nabla$ satisfying condition
$$\nabla_l c^i_{jk}=\nabla_j c^i_{lk},$$
where $c^i_{jk}$ is the $(1,2)$ tensor field representing the product $\circ$.  
\end{de}

Using commutativity of the algebra, it is easy to check  that in flat coordinates we have
$$c^i_{jk}=\d_j\d_k C^i.$$
In the Frobenius case, due to the existence of an invariant metric, one can make an additional
 step and obtain $C^i=\eta^{ij}\d_j F$ for a suitable function $F$ (the Frobenius potential).
\newline
\newline 
Given an $F$-manifold with compatible connection one can define the associated principal hierarchy in the following way. First, using a frame of flat vector fields $(X_{(1,0)}, \dots, X_{(n,0)})$, one defines the so called \emph{primary flows} by means of \beq\label{primary}u^i_{t_{(p,0)}}=c^i_{jk}X^k_{(p,0)}u^j_x,\quad p=1, \dots, n\eeq
then, using the recursive relations
\begin{equation}\label{recursive2}
\nabla_j X^i_{(p,\alpha)}=c^i_{jk}X^k_{(p,\alpha-1)}.
\end{equation}
one defines the \emph{higher flows}: 
\beq\label{higher}u^i_{t_{(p,\alpha)}}=c^i_{jk}X^k_{(p,\alpha)}u^j_x, \quad p=1, \dots, n.\eeq
The principal hierarchy associated to an $F$-manifold $(M, \circ)$ with compatible flat connection $\nabla$ is the collection of all the flows in \eqref{primary}, \eqref{higher}. 

Notice that the recursive relations \eqref{recursive2} can be written as
\beq\label{recursive3}[X_{(q,0)},X_{(p,\alpha)}]=X_{(p,\alpha-1)}\circ X_{(q,0)}.\eeq
The commutativity of the flows of the principal hierarchy can be proved using the following lemma \cite{LPR}
\begin{lemma}
The flows
\begin{equation*}
u^i_t=c^i_{jk}X_{(1)}^k u^j_x\qquad i=1,\dots,n.
\end{equation*}
and
\begin{equation*}
u^i_t=c^i_{jk}X_{(2)}^k u^j_x\qquad i=1,\dots,n.
\end{equation*}
associated with different solutions of
\begin{equation}\label{av}
c^i_{jm}\nabla_k X^m=c^i_{km}\nabla_j X^m.
\end{equation}
commute.
\end{lemma}
Indeed, the vector fields of the principal hierarchy satisfy \eqref{av}. For $\alpha=0$, it is trivial. For $\alpha>0$, it follows from the following chain of identities:
$$c^i_{jm}\nabla_k X^m_{(p,\alpha)}=c^i_{jm}c^m_{kl}X^l_{(p,\alpha-1)}=c^i_{km}c^m_{jl}X^l_{(p,\alpha-1)}=c^i_{km}\nabla_j X^m_{(p,\alpha)},$$
where the one in the middle is due to the associativity of $\circ$.

\section{Hamiltonian formalism for the principal hierarchy}
Using \eqref{recursive2}, it is immediate to see that all the flows of the principal hierarchy can be written as
$$u^i_{t_{(p,\alpha)}}=\nabla_j X^i_{(p,\alpha+1)} u^j_x.$$
Let $g$ be any metric compatible with $\nabla$  ($\nabla g=0$), $g$ not necessarily positive definite. Given the vector fields $X_{(p, \alpha+1)}$ define corresponding forms
$$(\omega_{(p, \alpha+1)})_l=g_{il}X^i_{(p, \alpha+1)}.$$
In this way we can write 
$$u^i_{t_{(p,\alpha)}}=\nabla_j X^i_{(p,\alpha+1)} u^j_x=\nabla_j\left(g^{il}(\omega_{(p, \alpha+1)})_l\right)u^j_x$$
$$=g^{il}\d_j(\omega_{(p,\alpha+1)})_l u^j_x-g^{im}\Gamma^l_{jm}u^j_x(\omega_{(p,\alpha+1)})_l=
\left(g^{il}\d_x+\Gamma^{il}_{j}u^j_x\right)(\omega_{(p,\alpha+1)})_l.$$
 Notice that the operator
 in the bracket is the differential operator associated
 with a Poisson bracket of hydrodynamic type.
\newline
\newline
Now two cases are possible:
\begin{itemize}
\item The metric $g$ is invariant with respect to the product, namely $g(X\circ Y, Z)=g(X, Y\circ Z)$, or equivalently $c^i_{jk}g^{kl}=c^l_{jk}g^{ki}$. 
\item The metric $g$ is not invariant with respect to the product.
\end{itemize}
The first case is less interesting since it is well known
 and well studied. Indeed in this case the 1-forms 
$\omega_{(p,\alpha+1)}$ are exact and we end up
 with the usual local Hamiltonian formalism introduced
 by Dubrovin and Novikov. To prove this fact we need the following lemma
\begin{lemma}
If $g$ is invariant then the 1-forms $\omega_{(p,\alpha)}$ defining the principal hierarchy
 satisfy the following recursive relations
\begin{equation}\label{recrel1form}
\nabla_j 
(\omega_{(p,\alpha)})_i=c^l_{ji}(\omega_{(p,\alpha-1)})_l.
\end{equation}
\end{lemma}

\n
\begin{proof}
The proof follows from the
 recursive relations for the vector fields. Indeed we can write
$$\nabla_j X^i_{(p,\alpha)}=c^i_{jk}X^k_{(p,\alpha-1)}$$
as
\begin{equation}\label{recrel1formgen}
g^{il}\nabla_j 
(\omega_{(p,\alpha)})_l=c^i_{jk}g^{kl}(\omega_{(p,\alpha-1)})_l.
\end{equation}
Using the invariance of $g$ with respect to $\circ$ we obtain
$$g^{il}\nabla_j 
(\omega_{(p,\alpha)})_l=c^l_{jk}g^{ki}(\omega_{(p,\alpha-1)})_l.$$
Multiplying both sides by $g_{hi}$ and taking the sum 
 over the index $i$ we obtain the result.
\end{proof}
\begin{corollary}
If $g$ is invariant with respect to $\circ$, then $\omega_{(p,\alpha)}$ are exact. 
\end{corollary}
\begin{proof}
Now using the above lemma, we obtain
$$(d\omega_{(p,\alpha)})_{ij}=\nabla_j(\omega_{(p,\alpha)})_i-\nabla_i(\omega_{(p,\alpha)})_j=0,$$
due to the commutativity of the product $\circ$. \end{proof}
On the other hand, if $g$ is not invariant with respect to the product $\circ$, in general the 1-forms $\omega$ are not exact since
\beq\label{exdefect}
(d\omega_{(p,\alpha)})_{ij}=\nabla_j(\omega_{(p,\alpha)})_i-\nabla_i(\omega_{(p,\alpha)})_j=
[g_{ki}c^k_{jl}-g_{kj}c^k_{il}]g^{lm}(\omega_{(p,\alpha-1)}))_m
\eeq
and the quantity in square brackets do not vanishes.
Although the $1$-forms defining the principal hierarchy starting from an $F$-manifold with compatible flat structure are not in general close, we can use the Poisson bracket on $1$-forms introduced in the previous Sections to reinterpret the commutativity of the flows of the principal hierarchy as the fact that the corresponding $1$-forms are in involution. 
This is the meaning of the following:
\begin{theorem}
The 1-forms $\omega_{(p,\alpha)}$ defining the principal hierarchy are in involution with respect 
 to the Poisson bracket \eqref{truePoisson1}, and the involutivity is a consequence of the associativity of the product $\circ$. 
\end{theorem} 

\n
\emph{Proof}. In flat coordinates $t^1, \dots, t^n$ the bracket \eqref{truePoisson1} reads
\begin{equation}\label{PB3}
\{\alpha,\beta\}=\int_{S^1}\eta^{il}\left(\f{\d\beta_l}{\d t^k}\f{\d\alpha_j}{\d t^i}-
\f{\d\alpha_l}{\d t^k}\f{\d\beta_j}{\d t^i}\right)t^k_x\,dx, 
\end{equation}
for $1$-forms whose coefficients depend only on the flat coordinates, not on their derivatives. 
Using the recursive relations \eqref{recrel1formgen}
 written in flat coordinates
\begin{equation}\label{recrel1formgenflat}
\eta^{il}\d_j 
(\omega_{(p,\alpha)})_l=c^i_{jk}\eta^{kl}(\omega_{(p,\alpha-1)})_l.
\end{equation}
we obtain
\begin{eqnarray*}
\{\omega_{(p,\alpha)},\omega_{(q,\beta)}\}_j&=&\\
&&\int_{S^1}\eta^{il}\left(\f{\d(\omega_{(q,\beta)})_l}{\d t^k}\f{\d(\omega_{(p,\alpha)})_j}{\d t^i}-
\f{\d(\omega_{(p,\alpha)})_l}{\d t^k}\f{\d(\omega_{(q,\beta)})_j}{\d t^i}\right)t^k_x\,dx=\\
&&\int_{S^1}\eta_{jm}(\omega_{(q,\beta-1)})_l(\omega_{(p,\alpha-1)})_s
\left(c^i_{kn}\eta^{nl}c^m_{ih}\eta^{hs}
-c^i_{kn}\eta^{ns}c^m_{ih}\eta^{hl}\right)t^k_x\,dx=\\
&&\int_{S^1}\eta_{jm}\eta^{nl}\eta^{hs}(\omega_{(q,\beta-1)})_l(\omega_{(p,\alpha-1)})_s
\left(c^i_{kn}c^m_{ih}
-c^i_{kh}c^m_{in}\right)t^k_x\,dx=0,
\end{eqnarray*}
due to the associativity of the product $\circ$. 
\endproof
Due to Proposition \ref{antih}, the above theorem provides an alternative proof of the commutativity of the flows of the principal hierarchy.

\begin{remark}
Notice that if $\alpha$ and $\beta$ are the differentials
 of two local functionals $H[u]=\int_{S^1}h(u)\,dx$ and $K[u]=\int_{S^1}k(u)\,dx$ respectively, then as was proved in Proposition \ref{usualpoisson} one has 
 $\{\delta H,\delta K\}=\delta\{H,K\} $
where $\{H,K\}$ is the usual Poisson brackets of hydrodynamic type:
$$\{H,K\}=\int_{S^1}\f{\d h}{\d t^l}\eta^{lm}\d_x\left(\f{\d k}{\d u^m}\right)\,dx.$$
 Therefore if $g$ is invariant the fact the 1-forms defining the principal hierarchy are in involution follows immediately from the involutivity of the corresponding
 Hamiltonians.
\end{remark}



\section{An example}
The triple $(\mathbb{R}^n,\nabla,\circ)$ with $\nabla$ defined by
\begin{equation}
\label{nat-conn-eps}
\begin{aligned}
&\Gamma^i_{jk}=0\qquad\mbox{for $i\ne j\ne k\ne i$}\\
&\Gamma^i_{jj}=-\Gamma^i_{ji}\qquad\mbox{for $i\ne j$}\\
&\Gamma^i_{ji}=\f{\epsilon}{u^i-u^j}\qquad\mbox{for $i\ne j$}\\
&\Gamma^i_{ii}=-\sum_{k\ne i}\Gamma^i_{ik}=-\sum_{k\ne i}\f{\epsilon}{u^i-u^k}\ .
\end{aligned}
\end{equation}
and $\circ$ defined by 
$c^i_{jk}=\delta^i_j\delta^i_k$ in coordinates $(u^1, \dots, u^n)$ (canonical coordinates) is an $F$-manifold with compatible flat structure \cite{LP}. It is strictly related to a dispersionless integrable hierarchy called the $\epsilon$-system. In the case 
 $\epsilon=1$ and for $n=3$ the flat coordinates $(t_1, t_2, t_3)$  are given in terms of canonical coordinates $(u_1, u_2, u_3)$ by the formulas:
\begin{eqnarray*}
t_1&=&u_1+u_2+u_3,\\
t_2&=&\f{1}{2(u_1-u_2)(u_3-u_1)},\\
t_3&=&\f{1}{2(u_1-u_2)(u_2-u_3)}.
\end{eqnarray*}
In flat coordinates any constant non degenerate symmetric matrix define a metric compatible with $\nabla$. For instance we can take the
 antidiagonal metric 
$$
g=\begin{pmatrix}
0 & 0 & 1\cr
0 & 1 & 0\cr
1 & 0 & 0          
\end{pmatrix}
$$ 
(this might not
 be the most convenient choice). At this point we have to decide if to work in canonical coordinates or in
 flat coordinates. In the first case the components of the metric
 become much more involved. However the 1-forms 
 definining the hierarchy can be easily obtained using
 the results of \cite{LP}. This becomes non trivial if we 
 work in flat coordinates since the 1-forms defining the hierarchy satisfy the system of PDEs
\begin{equation}\label{PDEform}
g^{km} c^i_{jk}\nabla_l \alpha_m=g^{km}c^i_{lk}\nabla_j \alpha_m
\end{equation}
involving the structure constants $c^i_{jk}$. 
\subsection{Working in canonical coordinates} 
The metric $g$ in canonical coordinates  reads:
\begin{eqnarray*}
g_{11}&=&\f{1}{4\left(u_{3}-u_{2}\right) \left(u_{1}-u_{2}\right) ^{4}\left(-u_{3}+u_{1}\right) ^{4}}
\left\{4{u_{1}}^{2}{u_{3}}^{4}-16{u_{1}}^{3}{u_{3}}^{3}+24{u_{1}}^{4}{u_{3}}^{2}-
16{u_{1}}^{5}u_{3}+\right.\\
&&\left.4{u_{1}}^{6}-
8u_{1}{u_{3}}^{4}u_{2}+32{u_{1}}^{2}{u_{3}}^{3}u_{2}-48u_{2}u_{3}^{2}{u_{1}}^{3}+32u_{2}{u_{1}}^{4}u_{3}-8u_{2}{u_{1}}^{5}+4{u_{2}}^{2}{u_{3}}^{4}+\right.\\
&&\left.-16u_{1}
{u_{2}}^{2}{u_{3}}^{3}+24{u_{1}}^{2}{u_{3}}^{2}{u_{2}}^{2}-16{u_{2}}^{2}u_{3}{u_{1}}^{3}+4{u_{2}}^{2}{u_{1}}^{4}+{u_{3}}^{2}u_{2}+{u_{3}}^{3}
-4{u_{3}}^{2}
u_{1}\right.\\
&&\left.-u_{3}{u_{2}}^{2}-4{u_{1}}^{2}u_{2}+4u_{3}
{u_{1}}^{2}+4u_{1}{u_{2}}^{2}-{u_{2}}^{3}\right\},
\end{eqnarray*}
\begin{eqnarray*}
g_{12}=g_{21}&=&\f{1}{4\left( u_{3}-u_{2}\right) ^{2} \left(u_{1}-u_{2} \right)^{4}\left( -u_{3}+u_{1}\right)^{3}}
\left\{2u_{1}u_{2}^{2}+u_{3}{u_{2}}^{2}-4u_{3}u_{1}u_{2}
+6{u_{1}}^{4}{u_{3}}^{2}+\right.\\
&&\left.-6{u_{1}}^{5}u_{3}-2u_{1}^{3}u_{3}
^{3}+u_{3}^{2}u_{2}+6u_{1}^{2}u_{2}^{3}u_{3}
-6u_{1}u_{2}^{3}u_{3}^{2}+18u_{1}
^{2}u_{3}^{2}u_{2}^{2}+6u_{1}^{2}u_{3}^{3}
u_{2}+\right.\\
&&\left.-6u_{1}u_{2}^{2}u_{3}^{3}+2u_{3}^{3}
u_{2}^{3}+2u_{3}^{2}u_{1}-18u_{2}^{2}u_{3}
u_{1}^{3}+18u_{2}u_{1}^{4}u_{3}-18u_{2}u_{3}
^{2}u_{1}^{3}+6u_{2}^{2}u_{1}^{4}+\right.\\
&&\left.-6u_{2}
u_{1}^{5}-u_{3}^{3}-u_{2}^{3}-2u_{1}^{3}u_{2}^{3}+2u_{1}^{6}\right\},
\end{eqnarray*}
\begin{eqnarray*}
g_{13}=g_{31}&=&\f{1}{4(u_{3}-u_{2})^{2}( u_{1}-u_{2})^3(-u_{3}+u_{1})^4}
\left\{-2u_{1}u_{2}^{2}-u_{3}u_{2}^{2}+2
u_{1}u_{3}^{4}u_{2}+4u_{3}u_{1}u_{2}+\right.\\
&&\left.+4u_{1}^{4}u_{3}^{2}-6u_{1}^{5}u_{3}+4u_{1}^{
3}u_{3}^{3}-u_{3}^{2}u_{2}+24u_{1}^{2}u_{3}^{
2}u_{2}^{2}+12u_{1}^{2}u_{3}^{3}u_{2}-16u_{1}
u_{2}^{2}u_{3}^{3}+\right.\\
&&\left.-2u_{3}^{2}u_{1}-16u_{2}
^{2}u_{3}u_{1}^{3}+22u_{2}u_{1}^{4}u_{3}-28
u_{2}u_{3}^{2}u_{1}^{3}-2u_{3}^{5}u_{2}
+4u_{2}^{2}u_{1}^{4}-6u_{2}u_{1}^{5}+\right.\\
&&\left.u_{3}^{3
}+u_{2}^{3}-6u_{1}^{2}u_{3}^{4}+2u_{1}u_{3}
^{5}+2u_{1}^{6}+4u_{2}^{2}u_{3}^{4}\right\},
\end{eqnarray*}
\begin{eqnarray*}
g_{22}&=&-\f{1}{4( u_{1}-u_{2})^{4}(u_{2}-u_{3}) ^{2}(-u_{3}+u_{1})^{2}}\left\{-4u_{3}^{2}u_{1}^{3}-
4u_{1}^{4}u_{3}+4u_{1}^{5}+24u_{3}u_{1}^{
3}u_{2}-16u_{1}^{4}u_{2}+\right.\\
&&\left.+12u_{1}u_{2}^{2}
u_{3}^{2}-36u_{1}^{2}u_{3}u_{2}^{2}+20u_{2}
^{2}u_{1}^{3}-8u_{2}^{3}u_{3}^{2}+16u_{1}
u_{2}^{3}u_{3}-8u_{1}^{2}u_{2}^{3}+4u_{1}^{2}
u_{3}^{3}-8u_{1}u_{3}^{3}u_{2}+\right.\\
&&\left.+4u_{2}^{2}
u_{3}^{3}-u_{2}^{2}+2u_{2}u_{3}
-u_{3}^{2}\right\},
\end{eqnarray*}
\begin{eqnarray*}
g_{23}=g_{32}&=&
\f{1}{4(u_{3}-u_{1})^{3}(u_{1}-u_{2})^{3}( u_{2}- u_{3})}\left\{2u_{3}^{3}u_{2}+u_{2}-6u_{3}^{2}
u_{1}u_{2}+
6u_{1}^{2}u_{2}u_{3}-2u_{1}^{3}u_{2}+\right.\\
&&\left.-u_{3}
-6u_{3}u_{1}^{3}-2u_{3}^{3}u_{1}+2u_{1}^{
4}+6u_{3}^{2}u_{1}^{2}\right\},
\end{eqnarray*}
\begin{eqnarray*}
g_{33}&=&\f{1}{4(u_{2}-u_{3}) ^{2}(u_{1}-u_{2}) ^{2}(-u_{3}+u_{1})^{4}}
\left\{4u_{1}u_{3}^{4}-16u_{1}^{2}u_{3}^{3}+24
u_{3}^{2}u_{1}^{3}-16u_{1}^{4}u_{3}+4u_{1}^{5}
+\right.\\
&&\left.-4u_{2}u_{3}^{4}+16u_{1}u_{3}^{3}u_{2}
-24u_{1}^{2}u_{3}^{2}u_{2}+16u_{3}u_{1}^{3}
u_{2}-4u_{1}^{4}u_{2}+u_{2}^{2}-2u_{2}u_{3}+ u_{3}^{2}\right\}.
\end{eqnarray*} 
The Poisson operator in canonical coordinates is given by the formula
$$g^{ij}\d_x+\Gamma^{ij}_ku^k_x$$
where $g^{ij}$ are the contravariant components
 of the euclidean metric in canonical coordinates
\begin{eqnarray*}
g^{11}&=&\f{4}{9}u_1^6-\f{4}{3}u_1^5u_{3}-\f{4}{3}u_{2}{u_{1}}^{5}+\f{14}{3}u_{2}u_{1}^{4}u_{3}+u_{1}^{4}u_{3}^{2}
+{u_{2}}^{2}{u_{1}}^{4}+\f{4}{9}
u_{1}^{3}{u_{3}}^{3}-4{u_{2}}^{2}u_{3}{u_{1}}^{
3}+\\
&&-\f{16}{3}u_{2}{u_{3}}^{2}{u_{1}}^{3}+
6{u_{1}}^{2}u_{3}^{2}{u_{2}}^{2}-\f{2}{3}u_{1}^{2}u_{3}^{4}+\f{4}{3}u_{1}^{2}u_{3}^{3}u_{2}-4u_{1}u_{2}^{2} u_{3}^{3}+\f{4}{3}u_{1}u_{3}^{4}u_{2}+\f{1}{9}u_{3}^{6}+\\
&&-\f{2}{3}
u_{3}^{5}u_{2}+\f{2}{9}u_{3}^{3}+u_{2}^{2}u_{3}^{4}-\f{2}{9}{u_{2}}^{3}-\f{2}{3}{u_{3}}^{2}u_{2}+\f{2}{3}u_{3}u_{2}^{2},
\end{eqnarray*}
\begin{eqnarray*}
g^{12}=g^{21}&=&-\f{2}{9}u_{1}^{6}+\f{1}{3}u_{2}u_{1}^{5}+u_{1}^{5}u_{3}-\f{5}{3}{u_{1}}^{4}{u_{3}}^{2}-\f{5}{3}u_{2}
 u_{1}^{4}u_{3}+{\frac {10}{9}}{u_{1}}^{3}{u_{3}}^{3}+\f{10}{3}u_{2}{u_{3}}^{2}{u_{1}}^{3}+\\
&&-\f{10}{3}{u_{1}}^{2}u_{3}^{3}u_{2}+\f{5}{3}u_{1}{u_{3}}^{4}u_{2}-\f{1}{3}u_{1}u_{2}^{2}+\f{2}{3}u_{3}u_{1}u_{2}-\f{1}{3}{u_{3}}^{2}{
u_{1}}-\f{1}{3}u_{1}u_{3}^{5}+\f{2}{9}{u_{3}}^{3}+\\
&&-\f{1}{3} 
u_{3}^{5}u_{2}-\f{1}{3}{u_{3}}^{2}u_{2}+\f{1}{9}{u_{2}}^{3}+\f{1}{9}
{u_{3}}^{6},
\end{eqnarray*}
\begin{eqnarray*}
g^{13}=g^{31}&=&-\f{2}{9}{u_{1}}^{6}+u_{2}{u_{1}}^{5}+\f{1}{3}u_{1}^{5}u_{3}+\f{2}{3}{u_{1}}^{4}{u_{3}}^{2}-3u_{2}u_{1}^{4}u_{3}-{u_{2}}^{2}{u_{1}}^{4}+4{u_{2}}^{2}u_{3}{u_{1}}^{3}+\\
&&+2u_{2}{u_{3}}^{2}{u_{1}}^{3}-{
\frac {14}{9}}{u_{1}}^{3}{u_{3}}^{3}-6{u_{1}}^{2}u_{3}^{2}{u_{2}}^{2}+\f{2}{3}u_{1}^{2}{u_{3}}^{4}+2{u_{1}}^{
2}{u_{3}}^{3}u_{2}+4u_{1}\,{u_{2}}^{2}{u_{3}}^{3}+\\
&&-\f{2}{3}
u_{3}u_{1}u_{2}+\f{1}{3}u_{1}{u_{3}}^{5}+\f{1}{3}u_{3}^{2}u_{1}-3u_{1}{u_{3}}^{4}u_{2}+\f{1}{3}u_{1}{u_{2}}^{2}-{u_{2}}^{2}{u_{3}}^{4}-\f{1}{9}{u_{3}}^{3}-\f{2}{9}
{u_{3}}^{6}+\\
&&{u_{3}}^{5}u_{2}-\f{2}{9}{u_{2}}^{3}+\f{1}{3}u_{3}{u_{2}}^{2},
\end{eqnarray*}
\begin{eqnarray*}
g^{22}&=&\f{1}{9}{u_{1}}^{6}-\f{2}{3} 
u_{1}^{5}u_{3}+\f{5}{3}{u_{1}}^{4}{u_{3}}^{2}-{\frac {20}{9}}u_{1}^{3}{u_{3}}^{3}+\f{5}{3}{u_{1}}^{2}{u_{3}}^{4}-\f{2}{3}u_{1}{u_{3}}^{5}-\f{2}{3}{u_{3}}^{2}u_{1}-\f{2}{3}u_{1}{u_{2}}^{2}+\\
&&+\f{4}{3}u_{3}u_{1}u_{2}+\f{2}{9}{u_{3}}^{3}+\f{4}{9}
{u_{2}}^{3}+\f{1}{9}{u_{3}}^{6}-\f{2}{3}u_{3}{u_{2}}^{2},
\end{eqnarray*}
\begin{eqnarray*}
g^{23}&=&g^{32}=\f{1}{9}
u_1^6-
\f{1}{3}{u_{1}}^{5}u_{3}-\f{1}{3}u_{2}{u_{1}}
^{5}+\f{5}{3}u_{2}{u_{1}}^{4}u_{3}-\f{10}{3}u_{2}{u_{3}}^
{2}{u_{1}}^{3}+{\frac {10}{9}}{u_{1}}^{3}{u_{3}}^{3}+\f{10}{3}
{u_{1}}^{2}{u_{3}}^{3}u_{2}+\\
&&-\f{5}{3}{u_{1}}^{2}{u_{3}}^{4}
+u_{1}{u_{3}}^{5}-\f{5}{3}u_{1}{u_{3}}^{4}u_{2}-\f{1}{9}{
u_{3}}^{3}+\f{1}{3}{u_{3}}^{2}u_{2}-\f{1}{3}u_{3}{u_{2}}^{2
}-\f{2}{9}{u_{3}}^{6}+\f{1}{9}{u_{2}}^{3}+\f{1}{3}{u_{3}}^{5}u_{2},
\end{eqnarray*}
\begin{eqnarray*}
g^{33}&=&\f{1}{9}{u_{1}}^{6}-\f{2}{3}u_{2}u_{1}^{5}+{u_{2}}^{2}{u_{1}}^{4}+\f{4}{3}u_{2}{u_{1}}^{4}u_{3}-\f{2}{3}{u_{1}}^{4}{u_{3}}^{2}-4{u_{2}}^{2}u_{3}u_{1}^{3}+\f{4}{9}{u_{1}}^{3}{u_{3}}^{3}+\f{4}{3}u_{2}u_{3}^{2}{u_{1}}^{3}+\\
&&+{u_{1}}^{2}{u_{3}}^{4}+6{u_{1}}^{2}u_{3}^{2}{u_{2}}^{2}-\f{16}{3}{u_{1}}^{2}{u_{3}}^{3}u_{2}-
\f{4}{3}u_{1}{u_{3}}^{5}+\f{2}{3}{u_{3}}^{2}u_{1}-4u_{1}
{u_{2}}^{2}{u_{3}}^{3}+\f{2}{3}u_{1}{u_{2}}^{2}+\\
&&\f{4}{3}u_{1}{u_{3}}^{4}u_{2}-\f{4}{3}u_{3}u_{1}u_{2}+\f{4}{9}
{u_{3}}^{6}-\f{4}{3}{u_{3}}^{5}u_{2}-\f{4}{9}{u_{3}}^{3}+\f{2}{3}
{u_{3}}^{2}u_{2}-\f{2}{9}{u_{2}}^{3}+{u_{2}}^{2}{u_{3}}^{4},
\end{eqnarray*}
and $\Gamma^{ij}_k$ are the Christoffel symbols
 \eqref{nat-conn-eps} with $\epsilon=1$. The 1-forms $\omega^{(p,\alpha)}$  definining the principal hierarchy can be obtained starting 
 from the the differentials of the Casimirs of 
$$P^{ij}=g^{ij}\delta'(x-y)-g^{il}\Gamma^{j}_{lk}u^k_x\delta(x-y)$$
and solving the recursive relations
\begin{equation}\label{recrelforms}
\nabla_k\omega^{(p,\alpha+1)}_h=g_{ih}c^i_{kl}g^{lm}\omega^{(p,\alpha)}_m.
\end{equation}
For instance if we take the counity $\omega^{(1,0)}=3dt_3$:
\begin{eqnarray*}
\omega^{(1,0)}_1&=&-\f{3}{2}\f{1}{(u_1-u_2)^{2}(u_2-u_3)}\\ 
\omega^{(1,0)}_2&=&-\f{3}{2}\f{u_1-2u_2+u_3}{(u_1-u_2)^{2}(u_2-u_3)^2}\\
\omega^{(1,0)}_3&=&\f{3}{2}\f{1}{(u_1-u_2)(u_2-u_3)^2},\\
\end{eqnarray*}
it is easy to check that the 1-form $\omega^{(1,1)}$ has components:
\begin{eqnarray*}
\omega^{(1,1)}_1&=&-\f{1}{4(u_1-u_3)^{3}(u_1-u_2)^{3}(u_2-u_3)}\left\{4u_2u_1^4+2u_3u_1^4-
14u_2u_3u_1^3-4u_2^2u_1^3+\right.\\
&&\left.-6u_3^2u_1^3+12u_2^2u_3u_1^2+18u_2u_3^2u_1^2+6u_3^3u_1^2+2u_1u_2
-12u_2^2u_3^2u_1-2u_3u_1+\right.\\
&&\left.-2u_1u_3^{4}-10u_2u_3^{3}u_1-u_2^2+u_3^{2}+4u_2^2u_3^3+2u_2u_3^4\right\},
\\ 
\omega^{(1,1)}_2&=&-\f{1}{4(u_1-u_3)^{2}(u_1-u_2)^{3}(u_2-u_3)^{2}}\left\{2u_1^{5}-6u_2u_1^{4}+2u_3u_1^{4}+2u_2u_3u_1^{3}
+2u_2^{2}u_1^{3}+\right.\\
&&\left.-8u_3^{2}u_1^{3}+12u_2u_3^{2}u_1^{2}+2u_3^{3}u_1^{2}+2u_2^{
3}u_1^{2}-6u_2u_3^{3}u_1-6u_2^{2}
u_3^{2}u_1-4u_2^{3}u_3u_1+\right.\\
&&\left.2u_1u_3^{4}+4u_2^{2}u_3^{3}-u_2^{2}+2u_2u_3
-u_3^{2}-2u_2u_3^{4}+2u_2^{3}u_3^{2}\right\},\\
\omega^{(1,1)}_3&=&\f{1}{4(u_1-u_3)^{3}(u_1-u_2)^{2}(u_2-u_3)^{2}}\left\{2u_1^{5}-6u_3u_1^{4}+2u_2
u_1^{4}-4u_2^{2}{u_1}^{3}+6u_3^{2}u_1^{3}+\right.\\
&&\left.-6u_2u_3u_1^{3}-2u_3^{3}u_1^{2}+12u_2^{2}u_3u_1^{2}+6u_2
u_3^{2}{u_1}^{2}-12u_2^{2}{u_3}^{2}u_1-2
u_2u_3^{3}u_1-2u_2u_3+\right.\\
&&\left.{u_3}^{2
}+{u_2}^{2}+4u_2^{2}{u_3}^{3}\right\}.
\end{eqnarray*} 
The associated flow is the first primary flow:
$$\f{\d u_i}{\d t}=\f{\d u_i}{\d x},\qquad i=1,\dots,3.$$
In the same way we obtain the components of the 1-form $\omega^{(1,2)}$:
\begin{eqnarray*}
\omega^{(1,2)}_1&=&\f{1}{4(u_1-u_3)^{3}(u_1-u_2)^{3}(u_2-u_3)}\left\{2u_1^{6}-6u_2u_1^{5}-6u_3
u_1^{5}+14u_2u_1^{4}u_3+4u_2^{2
}u_1^{4}+\right.\\
&&\left.6u_3^{2}u_1^{4}-6u_2u_1^{3}u_3^{2}-8u_2^{2}u_3u_1^{3}-2u_3^{3}u_1^{3}-6u_2u_1^{2}u_3^{3}-2
u_1u_2^{2}+2u_3^{2}u_1+8u_2^{2}
u_3^{3}u_1+\right.\\
&&\left.4u_2u_1u_3^{4}-{u_3
}^{3}+u_2^{3}-4u_2^{2}u_3^{4}-u_2u_3^{2}+u_3u_2^{2}\right\},\\ 
\omega^{(1,2)}_2&=&-\f{1}{4(u_1-u_3)^{2}(u_1-u_2)^{3}(u_2-u_3)^{2}}\left\{4u_3u_1^{5}-8u_2u_1^{4}u_3
-4u_3^{2}u_1^{4}-2u_2^{2}u_1^{4
}+2u_2^{3}u_1^{3}+\right.\\
&&\left.6u_2^{2}u_3{u_1
}^{3}+12u_2u_1^{3}u_3^{2}-4u_3^{3}u_1^{3}+4u_3^{4}u_1^{2}-6u_2^{2}u_3^{2}u_1^{2}-2u_2^{3}u_3u_1^{2}-4u_2u_1u_3^{4}+2u_2^{2}u_3^{3}
u_1+\right.\\
&&\left.-2u_2^{3}u_3^{2}u_1-u_2^{3}+u_2u_3^{2}-u_3^{3}+2u_2^{3}u_3^{3}+u_3
u_2^{2}\right\},\\
\omega^{(1,2)}_3&=&\f{1}{2(u_1-u_3)^{3}(u_1-u_2)^{2}(u_2-u_3)^{2}}\left\{2u_1^{5}-6u_3u_1^{4}+2u_2
u_1^{4}-4u_2^{2}u_1^{3}+6u_3^{2}u_1^{3}-6u_2u_3u_1^{3}+\right.\\
&&\left.
-2u_3^{3}u_1^{2}+12u_2^{2}u_3u_1^{2}+6u_2u_3^{2}u_1^{2}-12u_2^{2}u_3^{2}u_1-2
u_2u_3^{3}u_1-2u_2u_3+u_3^{2
}+u_2^{2}+4u_2^{2}u_3^{3}\right\}.\\
\end{eqnarray*}
and the associated flow 
\begin{eqnarray*}
\f{\d u_1}{\d t}&=&(u_2+u_3)\f{\d u_1}{\d x},\\
\f{\d u_2}{\d t}&=&(u_1+u_3)\f{\d u_2}{\d x},\\
\f{\d u_3}{\d t}&=&(u_1+u_2)\f{\d u_3}{\d x}.
\end{eqnarray*}
Higher flows can be obtained analogously.
\begin{remark}
As it can be observed from the example above, even if the starting points of the hierarchy are the differentials of the Casimirs of $P$, hence exact $1$-forms,  the iterative procedure leads immediately to consider non exact 1-forms. Indeed,  in the example $\omega^{(1,1)}$ is exact
 and $\omega^{(1,2)}$ is not exact. The reason is that
 in the first case ($p=1,\,\alpha=1$) the r.h.s. of \eqref{exdefect} vanishes while for different values of $p$ and $\alpha$, including $p=1,\,\alpha=2$  it does not vanish. 
\end{remark}


\begin{thebibliography}{99}

\bibitem{AL1} 
A. Arsie and P. Lorenzoni \emph{$F$-manifolds with eventual identities, bidifferential calculus and twisted Lenard-Magri chains},
 arXiv:1110.2461, to appear in Int. Math. Res. Notices. 

\bibitem{AL2} 
A. Arsie and P. Lorenzoni \emph{From Darboux-Egorov system to bi-flat $F$-manifolds},
 arXiv:1205.2468, submitted. 
 
\bibitem{DT} P. Dedecker and W.M. Tulczyjev, {\em Spectral sequences and the inverse problem of
the calculus of variations}, Lecture Notes in Math. 836 (1980) 498-503.
 
\bibitem{DN} B.A. Dubrovin, S.P. Novikov,
 \emph{On Hamiltonian brackets of hydrodynamic type},
 Soviet Math. Dokl. {\bf 279:2} (1984) 294--297.

\bibitem{du93}
B.A. Dubrovin, \emph{Geometry of 2D topological field theories},
in: Integrable Systems and Quantum Groups, Montecatini Terme, 1993.
Editors: M. Francaviglia, S. Greco. Springer Lecture Notes in
Math. {\bf 1620} (1996), pp.\ 120--348.


\bibitem{DZ} B. Dubrovin, Y. Zhang, \emph{Normal forms of integrable
PDEs, Frobenius manifolds and Gromov-Witten invariants},   math.DG/0108160.


\bibitem{fe91} E.V. Ferapontov,
\emph{Differential geometry of nonlocal Hamiltonian operators
 of hydrodynamic type}, Funct. Anal. Appl. {\bf 25}  (1991),
  no. 3, 195--204 (1992).

\bibitem{fm90} E.V. Ferapontov, O.I. Mokhov,
\emph{Nonlocal Hamiltonian operators of hydrodynamic type that
 are connected with metrics of constant curvature},
  Russ. Math. Surv.  {\bf 45}  (1990),  no. 3, 218--219


\bibitem{G} E. Getzler, \emph{A Darboux theorem for Hamiltonian operators
in the formal calculus of variations}, Duke Math. J. {\bf 111} (2002) 535--560.

\bibitem{GD}
I.M. Gel'fand and I.A. Dorfman, \emph{Hamiltonian operator and the classical Yang-Baxter equation}, Funct. Anal. Appl, 16 (1982), 241--248.

\bibitem{HM}
C. Hertling, Y. Manin, \emph{Weak Frobenius manifolds}, Internat. Math. Res. Notices {\bf 1999}, no. 6, 277--286.

\bibitem{LZ}
S.Q. Liu, Y. Zhang, \emph{Jacobi structures of evolutionary partial differential equations}, Adv. Math. 227 (2011), no. 1, 73--130.

\bibitem{LPR}
P. Lorenzoni, M. Pedroni, A. Raimondo, \emph{$F$-manifolds and integrable systems
 of hydrodynamic type}, Archivum Mathematicum 
{\bf 47} (2011), 163-180.

\bibitem{LP}
P. Lorenzoni, M. Pedroni, \emph{Natural connections for semi-Hamiltonian systems: 
The case of the $\epsilon$-system}, Letters in Mathematical Physics, {\bf 97} (2011), no. 1, 85--108.


\bibitem{manin} Y. Manin,
\emph{$F$-manifolds with flat structure and Dubrovin's duality},
Adv. Math. {\bf 198} (2005), no. 1, 5--26.

\bibitem{MM}
F. Magri and C. Morosi, \emph{A geometrical characterization of integrable Hamiltonian systems through the theory of Poisson-Nijenhuis manifolds},
 Quaderno S19 (1984) (University of Milan).

\bibitem{O} P. Olver, {\em Applications of Lie groups to differential equations}, Springer-Verlag GTM, second edition, 2000.

\bibitem{S} I.A.B.Strachan, \emph{Deformations of the Monge/Riemann hierarchy and approximately integrable systems},
 J. Math. Phys. {\bf 44} (2003) 251--262.

\bibitem{ts} S.P. Tsarev,
\emph{The geometry of Hamiltonian systems of hydrodynamic type. The
generalised hodograph transform},
USSR Izv. {\bf 37} (1991) 397--419.


\end{thebibliography}
\end{document}